\documentclass[
    aps,
    prx,
    twocolumn,
    superscriptaddress,
    nofootinbib,
    longbibliography
]{revtex4-2}

\usepackage{mathtools}
\usepackage{amssymb}
\usepackage{amsthm}
\usepackage{mathrsfs}
\usepackage{bm,dsfont}
\usepackage{enumitem}
\usepackage{physics}
\usepackage{graphicx}
\usepackage{xcolor}

\usepackage[expansion=false]{microtype}
\usepackage{natbib}
\usepackage{tikz}
\usepackage[normalem]{ulem}
\usepackage{booktabs}
\usetikzlibrary{arrows.meta}

\usetikzlibrary{arrows.meta,graphs,graphs.standard,positioning}

\definecolor{paperBlue}{HTML}{33658A}
\definecolor{paperTeal}{HTML}{2A9D8F}
\definecolor{paperOrange}{HTML}{E07A5F}
\definecolor{paperInk}{HTML}{273746}
\definecolor{paperGray}{HTML}{6C757D}

\usepackage[
colorlinks=true,
linkcolor=paperBlue,
citecolor=paperBlue,
urlcolor=paperTeal,
pdfauthor={Riccardo Castellano, Pavel Sekatski, and Nicolas Brunner},
pdftitle={Exact classical simulation of finite-dimensional prepare-and-measure experiments with finite expected message length}
]{hyperref}
\usepackage{cleveref}
\usepackage{hypcap}

\newtheorem{theorem}{Theorem}
\newtheorem{lemma}[theorem]{Lemma}
\newtheorem{proposition}[theorem]{Proposition}

\newtheorem{definition}{Definition}

\newcommand{\ii}{\mathrm{i}}

\newcommand{\Pure}{\mathbb{P}(\mathbb{C}^{d})}
\newcommand{\Prob}{{\rm Pr}}
\newcommand{\Exp}{\mathbb{E}}

\newcommand{\dmu}{\,\mathrm{d}}

\newcommand{\one}{\mathds{1}}
\DeclareMathOperator{\diam}{diam}

\newcommand{\cC}{\mathcal{C}}

\newcommand{\cP}{\mathcal{P}}

\newcommand{\cF}{\mathcal{F}}

\newcommand{\AlpS}{l}
\newcommand{\Alp}[1]{[\AlpS_{#1}]}

\begin{document}

\title{The cost of classically simulating quantum communication is as high as it could be}

\author{Riccardo Castellano}
\affiliation{Department of Applied Physics, University of Geneva, 1211 Geneva, Switzerland}
\author{Pavel Sekatski}
\affiliation{Department of Applied Physics, University of Geneva, 1211 Geneva, Switzerland}
\author{Nicolas Brunner}
\affiliation{Department of Applied Physics, University of Geneva, 1211 Geneva, Switzerland}

\begin{abstract}
A long-standing open question is to characterize the classical communication cost of simulating the correlations arising from measurements on finite-dimensional quantum systems. We consider this question in the prepare-and-measure scenario. First we construct a classical simulation model involving a finite average number of bits of one-way communication that simulates exactly the correlations arising from arbitrary $d$-dimensional quantum states and arbitrary measurements. Second, we show that no classical simulation protocol involving a finite average message dimension can exist for $d\geq 4$, hence strengthening a recent result of de Gois, Santos and Vieira. Our first result relies on a more general simulation theorem which can be applied to certain continuous variable (i.e. infinite-dimensional) quantum systems, and even beyond quantum theory, e.g. to generalized probabilistic theories. In particular, we show that multimode Gaussian states of bounded average energy and general Gaussian measurements admit an exact classical simulation with a finite expected number of bits. Overall, our results show that the classical simulation cost of quantum communication is essentially as high as it could be.

\end{abstract}

\maketitle

\section{Introduction}
Quantum physics features a number of phenomena that evade an explanation within classical physics. Key examples include the violation of Bell inequalities, the possibility to guarantee secrecy of communication based on physical principles, and speed-up in computing. 

To get a deeper understanding of these effects, it is interesting to quantify the advantage offered by quantum resources over classical ones. While being a central question in the fields of quantum computing and simulation, this question has also been considered in the case of Bell nonlocality, see e.g. \cite{Buhrman2010,Brunner2014}. Specifically, one considers classical models based on shared randomness (local hidden variables) augmented by some extra classical communication, and then asks how much classical communication is required for reproducing exactly some given quantum Bell nonlocal correlations \cite{Maudlin1992,Brassard1999,Steiner2000,Cerf2000,Pironio2003,Montina2016}. 

This question has been initially investigated for the simplest case of a maximally entangled two-qubit (Bell) state, considering arbitrary local measurements performed by two distant parties Alice and Bob. 
%Both worst-case and average communication have been considered. 
In 2003, Toner and Bacon \cite{TonerBacon2003} strikingly demonstrated that a single bit of classical communication 
%(in the worst case) 
is enough to simulate exactly the correlations of a Bell state under arbitrary projective measurements. This result was recently extended to include arbitrary local measurements (i.e. POVMs), in which case two bits of communication are enough \cite{RennerTavakoliQuintino2023}. 

This naturally raises the question of determining the communication cost of other entangled states \cite{TonerBacon2003,Renner2023}, in particular moving to the case of higher dimensions, i.e. qudit entanglement \cite{Degorre2007,Regev2010,Vertesi2009}. Establishing connections with the field of communication complexity \cite{Buhrman2010}, it was shown that the simulation of correlations arising from local measurements on entangled states of local dimension $d$ requires classical communication with alphabet-size of at least $2^{\Omega(d)}$  \cite{Brassard1999}, see also \cite{gavinsky2008,Perry2015,Montanaro2019,Havlicek2020,YCL2026}. 

Equivalently, these questions can be translated to a prepare-and-measure scenario, where Alice prepares an arbitrary qudit state and sends it to Bob who performs a measurement \cite{Cerf2000,brask2026}. The question is then how much one-way classical communication is required to exactly simulate quantum correlations arising from qudits \cite{Gallego2010,schlosser2026}. Similarly to the Bell scenario, results from communication complexity imply an exponential separation between classical and quantum resources, even when allowing for unlimited shared randomness. Specifically, in order to simulate the correlations obtained via a quantum message of dimension $d$ (i.e. a qudit), one requires a classical message of alphabet-size of at least $2^{\Omega(d)}$, see e.g. \cite{Buhrman2003,brask2026}.

While these results provide relevant lower bounds, they naturally lead to the question of what is the actual classical communication cost of qudit correlations, and even more interestingly, whether this simulation is even possible with finite classical communication?

Recently, de Gois, Santos and Vieira \cite{deGoisSantosVieira2026} reported significant progress in this direction. They proved that no classical communication model involving finite (but fixed) alphabet-size can exactly simulate quantum correlations for $d\geq 4$, meaning that qudit correlations cannot be simulated with finite classical communication in the worst-case scenario. This leaves open the possibility of a classical simulation using a finite \emph{expected} amount of communication.

In this work, we address this question. First, we show that qudit correlations cannot be simulated by any classical communication model with finite expected alphabet-size, hence strengthening the result of \cite{deGoisSantosVieira2026}. Second, we demonstrate that all qudit correlations, for arbitrary states and measurements, can be classically simulated with a protocol involving a finite expected number of bits $O(d \log(d))$ of one-way classical communication. Importantly, as the protocol involves only one-way communication (from preparation to measurement), it respects the causal structure of the scenario; in contrast, previous works have presented simulation models relying on two-way communication \cite{Massar2001,Brassard2018,Montanaro2019}, which are arguably more costly and less natural in this context.

Overall, our results identify (up to the exact value of an exponent) the communication complexity of simulating the correlations of any finite-dimensional quantum systems in the prepare-and-measure scenario. From our model for exact simulation, we also obtain a classical model for approximate simulation of qudit correlations with finite number of bits, with exponentially decaying error. Moreover, these results also have implications on the communication cost of simulating qudit correlations in the bipartite Bell scenario. In particular, from our second result it follows that the correlations of every entangled state (of dimension $d \times d$) and arbitrary local measurements can be simulated with one-way classical communication of $O(d \log(d))$ bits on average. We also present a model

Our second result is based on a general simulation theorem which applies to broad classes of theories, formalized as probability-valued maps defined on general state spaces. Notably, this theorem can be used to construct classical communication models with finite expected number of bits simulating infinite-dimensional quantum systems, or more general theories, e.g. generalized probabilistic theories \cite{Barrett2007,Janotta2014}. 
In particular, we present such a classical simulation model for a class of continuous variable (CV) systems, involving infinite-dimensional quantum states and measurements with continuous outcomes. Specifically, we show that multimode Gaussian states (pure or mixed) with bounded energy and arbitrary general-dyne measurements \cite{Weedbrook2012}, can be classically simulated using an expected number of bits that grows logarithmically with the energy.

\section{Scenario and main results}\label{sec:general-framework}

We will present our results focusing on the prepare-and-measure scenario, and discuss in turn how they apply to the bipartite Bell scenario. A prepare-and-measure scenario features two distant parties. Alice holds the preparation device, which will produce a physical system, then transmitted to Bob, who holds the measurement device, which will perform a measurement on the system and return a measurement outcome. 

In a quantum prepare-and-measure experiment, Alice will be asked to prepare a certain quantum state $\psi$ chosen among the continuous set of all $d$-dimensional pure states, i.e. 
\begin{equation}
    \Pure:=\{\rho_\psi=\ketbra{\psi}:\psi\in\mathbb C^d,
    \ \braket{\psi}=1\}.
\end{equation}
Bob will be asked to perform a certain quantum measurement, represented by a POVM $\mathsf E:=\{E_{b}\}_{i=b,\dots s}$, whose elements (effects) satisfy 
\begin{equation}
    \mathsf E=\{E_b\}_{b=1}^s,
    \qquad E_b\geq0,
    \qquad \sum_{b=1}^s E_b=\one,
\end{equation}
and $b\in \{0,..,s\}$ denotes the measurement outcome. The observed input/output statistics of this experiment are given by the Born rule:  
\begin{equation} \label{eq:Born-rule}
    p(b \mid\mathsf E,\psi):=
    \operatorname{tr}[E_{b} \,\rho_\psi].
\end{equation}
Note that here we can restrict our attention to extremal POVMs\footnote{Indeed any POVM acting on $d$-dimensional quantum states can be decomposed into POVMs with at most $d^2$ outcomes followed by classical post-processing.} which have output cardinality $s\leq d^{2}$ \cite{DArianoLoPrestiPerinotti2005,ChiribellaDArianoSchlingemann2007}.

Our main question of interest is now to simulate the quantum statistics, via a purely classical model, see Fig. 1. As above, Alice and Bob receive a classical description of a quantum state and of a quantum measurement, respectively, and their goal is to devise a purely classical strategy for sampling exactly from the quantum distribution \eqref{eq:Born-rule}. The simulation model is based on two resources. First, Alice and Bob have access to a pre-shared random classical variable $\lambda$. Second, Alice sends a classical message to Bob, which depends on the input state $\psi$ and $\lambda$. The goal is to understand how much classical communication is required for simulating quantum correlations arising from qudits.

In order to state our main results, let us first define more formally the elements of a classical simulation model.

%{\color{red} \it [it's a bit annoying that $m$ denotes the message sent, and $m_\lambda$ denotes the alphabet size. Maybe we can use the standard (e.g. Clover\&Thomas) $\mathcal{X}_\lambda=\{1,\dots,|\mathcal{X}_\lambda|\}$  and $|\mathcal{X}_\lambda|$ for the alphabet size?]}

\begin{figure}[t!]
    \centering
    \resizebox{\linewidth}{!}{%
    \begin{tikzpicture}[
        box/.style={
            draw, rounded corners=2pt,
            minimum width=1.05cm,
            minimum height=.7cm
        },
        arrow/.style={-{Latex[length=2mm]}, thick},
        shared/.style={arrow, dashed}
    ]

    % Quantum communication
    \node[font=\bfseries] at (1.35,1.65) {(a) Quantum};

    \node[box] (AQ) at (0,0) {Alice};
    \node[box] (BQ) at (2.7,0) {Bob};

    \node (psiQ) at (0,1) {$\psi$};
    \node (EQ) at (2.7,1) {$\mathsf E$};

    \draw[arrow, double, double distance=1pt] (psiQ) -- (AQ);
    \draw[arrow, double, double distance=1pt] (EQ) -- (BQ);
    \draw[arrow] (AQ) --
        node[below] {$\lvert\psi\rangle$} (BQ);
    \draw[arrow, double, double distance=1pt] (BQ.south) -- ++(0,-.55)
        node[below] {$b$};

    % Classical simulation
    \node[font=\bfseries] at (6.95,1.65) {(b) Classical};

    \node[box] (AC) at (5.6,0) {Alice};
    \node[box] (BC) at (8.3,0) {Bob};

    \node (psiC) at (5.6,1) {$\psi$};
    \node (EC) at (8.3,1) {$\mathsf E$};
    \node[draw, circle] (L) at (6.95,1.05) {$\lambda$};

    \draw[arrow, double, double distance=1pt] (psiC) -- (AC);
    \draw[arrow, double, double distance=1pt] (EC) -- (BC);
    \draw[arrow, double, double distance=1pt] (L.south west) -- (AC.north east);
    \draw[arrow, double, double distance=1pt] (L.south east) -- (BC.north west);

    \draw[arrow, double, double distance=1pt] (AC) --
        node[below] {$m\in \Alp{\lambda}$} (BC);
    \draw[arrow, double, double distance=1pt] (BC.south) -- ++(0,-.55)
        node[below] {$b$};

    \end{tikzpicture}%
    }
    \caption{Quantum communication {\bf(a)} and classical simulation
    {\bf (b)}. Double arrows denote classical systems while single arrows denote quantum ones. In the quantum case, Alice prepares and sends a quantum system in a state $\ket{\psi} \in \mathbb{C}^d$, on which Bob performs a measurement (represented by a POVM $\mathsf E$) to produce the output $b$.
    In the classical simulation, Alice encodes $\psi$ in a classical message $m=M_\lambda(\psi)\in \Alp{\lambda§}=\{1,\dots,\AlpS_\lambda \}$, and
    Bob samples $b$ from the decoder
    $q(\cdot\mid\mathsf E,\lambda,m)$. Encoding and decoding also use the shared random variable $\lambda$, which in our model determines the alphabet size $\AlpS_\lambda$ of the classical message. Our main result is to construct a classical simulation with finite expected number of communicated bits $\Exp[\log \AlpS_\lambda]=O(d \log(d))$.}
    \label{fig:quantum-classical-simulation}
\end{figure}
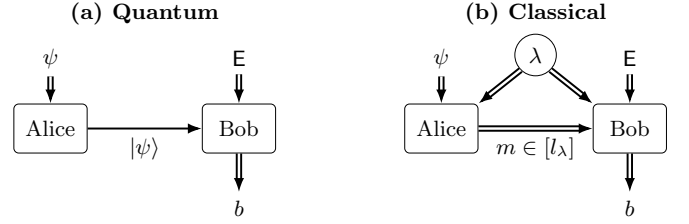

\begin{definition}[Classical one-way simulation]\label{def:classicalone-waySimul}
The parties share a random variable $\lambda$ on a probability
space $(\Omega,\mu)$, independent of their inputs. An encoder is a map
\begin{equation}
    M_\lambda:\Pure\longrightarrow \Alp{\lambda}:=\{1,\dots, \AlpS_\lambda\},
\end{equation}
where $\AlpS_\lambda\in\mathbb N$ is positive. A decoder assigns
a probability distribution $q(\cdot\mid\mathsf E,\lambda,m)$ on the
measurement outcomes. For discrete outcomes, this means
\begin{equation}
    q(b\mid\mathsf E,\lambda,m)\geq0,
    \qquad\sum_{b=1}^s q(b\mid\mathsf E,\lambda,m)=1.
    \label{eq:decoder-simplex}
\end{equation}
The encoder is independent of $\mathsf E$. The simulation is exact when
\begin{equation}
   \int_\Omega
    q\bigl(b\mid\mathsf E,\lambda,M_\lambda(\psi)\bigr)
    \,\mathrm d\mu(\lambda) =  p(b\mid\mathsf E,\psi)
    \label{eq:exact-simulation-general}
\end{equation}
for every preparation, measurement, and outcome $b$.
\end{definition}

Note that the encoder can in general be considered deterministic. This is because any randomness can be absorbed in $\lambda$.

The crucial resource we aim to quantify in these models is the amount of classical communication that is required. Three different quantities %, defined in Table~\ref{tab: tab1} 
can be considered:
\begin{itemize} \item the worst-case alphabet size 
\begin{equation}\label{eq:worstcase-cost}
\max[ \AlpS_\lambda] = \max_{\lambda \in \Omega} \AlpS_\lambda,
\end{equation}
\item the expected alphabet size  
 \begin{equation}\label{eq:expected-alphabet-size}
   \Exp[\AlpS_\lambda]=\int_\Omega \AlpS_\lambda
     \,\mathrm d\mu(\lambda),
 \end{equation}
\item the expected number of bits of communication \begin{equation}\label{eq:expected-bit-cost}
     \Exp[\log \AlpS_\lambda] = \int_\Omega \log \AlpS_\lambda
     \,\mathrm d\mu(\lambda),
 \end{equation}
 with $\log$ in base 2.
\end{itemize}

For the worst-case scenario, it was shown that $\max [\AlpS_\lambda]=2$ for qubits ($d=2$) \cite{Renner2023}, and more recently that $\max [\AlpS_\lambda] \leq 2^{357}$ for qutrits \cite{deGoisSantosVieira2026}. Astonishingly, for $d \geq 4$, it turns out that no classical simulation model using finite worst-case message dimension exists, i.e. $\max [\AlpS_\lambda]= \infty$ \cite{deGoisSantosVieira2026}. Here we start by deriving a strengthening of this result, proving a no-go result for the simulation with finite expected alphabet size.

\begin{theorem}[No finite average message dimension]
\label{thm:no-finite-average-classical-dimension}
For Hilbert space dimension $d\geq4$, every classical simulation protocol involving one-way communication reproducing exactly qudit correlations, as in definition \ref{def:classicalone-waySimul},
requires infinite expected alphabet size, i.e.
\begin{equation}
\Exp[\AlpS_\lambda]=\infty.
\label{eq:infinite-average-classical-dimension}
\end{equation}
Here we we require the encoder
to be jointly measurable in $(\lambda,\psi)$ and the decoder to be
jointly measurable in $(\lambda,\phi)$ for each message label \footnote{We use
the given $\sigma$-algebra on the shared-randomness space, the
Borel $\sigma$-algebra on $\Pure$, equipped with Fubini--Study measure,
and the discrete $\sigma$-algebra on the message labels. However the measurability assumption can likely be removed.}. 
\end{theorem}

The proof of this result is presented in section \ref{sec:finite-average-dimension-impossible} of the Appendix. The proof is essentially based on the original proof of theorem B.1 of \cite{deGoisSantosVieira2026}. 

Given the above result, the next natural question is whether there could exist a classical simulation model with a finite number of bits of communication. Note that this is in principle possible, since there exist models where the quantity $\Exp[\log \AlpS_\lambda]$ remains finite although $\Exp[\AlpS_\lambda]$ is infinite. Our second result is to show that such models exist.

\begin{theorem}[Exact simulation in every finite dimension]
\label{thm:main-simulation}
Every finite-dimensional quantum system admits a one-way
exact simulation for all states and POVMs with finite expected
communication. In particular, for $d\geq2$, one can achieve
\begin{equation}
    \Exp[\log \AlpS_\lambda]\leq
    \log d+n_d\left(\left\lceil\log(4d\sqrt{n_d})\right\rceil+2\right),
\end{equation}
where $d$ is the Hilbert-space dimension and $n_d:=2d-2$ is the
real dimension of the pure-state space. Furthermore, the probability of using more than $k$ bits of communication vanishes exponentially in $k$:
% \begin{equation} \label{eq:quantum-bit-tailth}
% \Prob\!\left(\ell(\lambda)>k\right)\leq C_{d}\,2^{-k/n_d}.
% \end{equation}
\begin{equation} \label{eq:quantum-bit-tailth}
\Prob\!\left(\log \AlpS_\lambda \geq k\right)\leq C_{d}\,2^{-k/n_d}.
\end{equation}

For some $C_{d}>0$.
\end{theorem}

The proof of this result follows from a more general simulation theorem that we present and prove in the next section. In a nutshell, this theorem provides the construction of a general simulation model using a finite average number of bits of classical communication, which can be applied to a broad class of physical theories, and not specific to quantum theory in any sense. Essentially, the theorem applies to any theory where the state space resembles a finite-dimensional metric space, where the metric is derived from the function assigning measurement outcome probabilities. A notable example is finite-dimensional generalized probabilistic theories \cite{Barrett2007,Janotta2014}. Moreover, the theorem can also be applied in certain scenarios involving infinite-dimensional quantum systems, where measurements have continuous outcomes. As an illustrative example, in Sec.~\ref{sec:gaussian} we present classical simulation models for a class of continuous variable quantum systems, namely multimode Gaussian states with bounded energy
subjected to arbitrary Gaussian (general-dyne) measurements. In this case, the model uses an average number of bits of communication that grows logarithmically with the energy.

The results from the two Theorems above also have implications on the classical communication cost of simulating qudit entanglement in the bipartite Bell scenario. In particular, from Theorem 3 it follows that the correlations resulting from arbitrary local POVMs on every entangled state (of dimension $d \times d$) can be simulated with one-way classical communication of $O(d)$ bits on average. The protocol works as follows. First, Alice locally samples her outcome from the quantum marginal distribution. Then she simulates the conditional quantum state for Bob via the protocol of Theorem 3; see \cite{Buhrman2010,brask2026} for more details. Additionally, from Theorem 2 it follows that the correlations of the maximally entangled state ($d \times d$, with $d \geq 4$) cannot be simulated with one-way communication of finite average message dimension; see e.g. Appendix C2 of \cite{deGoisSantosVieira2026} for more details.

Finally, coming back to the simulation of qudit correlations, we see that Theorem 2 can be viewed as a very specific application of the general simulation theorem. Combined with the result of Theorem 1, one can then argue that qudit correlations are essentially ``as hard to simulate as they can be''. In the sense that simulability with finite average message dimension is not possible, while simulability with finite average number of bits appears to be generic for a broad class of physical theories.

\section{General simulation theorem}
\label{sec:abstract-simulation}

This section is devoted to our general simulation theorem. Before stating the result and the proof, we provide a high-level (and hopefully intuitive) overview. We formulate the problem abstractly as follows. 

Alice is given an input belonging to the input set $X$, a set without any particular structure. Similarly, Bob's measurement is described 
by a function $f\in\cF$ that maps every state $x \in X$ into a probability distribution $f(x)$
over finitely many outcomes. In the simulation protocol, Alice gets as input state $x \in X$ while Bob receives an input measurement $f \in \cF$, and the two can use arbitrary classical shared randomness. In order to obtain a simulation model with finite expected classical communication, we need to impose restrictions on the pair $X$ and $\cF$, e.g. to exclude models with infinite classical capacity. 

We find that the key condition
for enabling exact simulation with finite expected communication is that $X$ can be covered by finitely many arbitrary ``small'' subsets.   In the absence of additional structure on $X$ (such as a metric or a topology), to make sense of ``smallness'' of a subset $R\subset X$ we define the notion of the mass
deficit, which is directly induced by the functions $\cF$. This definition is based on the following idea: a subset $R\subset X$ is ``small'' if, regardless of which function $f\in\cF$ Bob is asked to simulate, the distributions $\{f(x):x\in R\}$ have a large amount of probability mass in common, i.e, $f(x)\sim f(y)$ for all $x,y\in R$. Suppose that $X$ can be
partitioned into finitely many  small subsets, and that Alice and Bob agree in advance on a partition of $X$ and label its  subsets.
Alice can then send Bob the label of the subset containing her input.
Bob thereby obtains a good approximation of the probability distribution
he is supposed to output. Finer partitions allow Bob to obtain a better
approximation of $f(x)$, but at the expense of longer messages. The
shared random variable $\lambda$ determines how fine the partition used in each round is, and thus how many bits Alice sends to Bob.

Consequently, a single quantitative requirement emerges: we require the
number of  subsets of size  (mass deficit) $\epsilon$ needed to cover $X$ to grow at most
polynomially in $\epsilon^{-1}$ as $\epsilon\to0$, ensuring that the
protocol uses a finite expected number of bits. This condition closely
resembles the notion of Minkowski dimension in metric spaces, the
difference being that our notion of the  size of a subset is induced by $\cF$ and not
necessarily derived from a metric.  In consequence, the condition enabling classical simulation with finite expected communication is a property of the pair $X,\cF$, and not of the  input set alone. 

\subsection{Inputs, target maps, and exact simulation}

Let $X$ be an arbitrary nonempty input set, and let
$s\in\mathbb N\cup\{\infty\}$. Define
\begin{equation}
    \Delta_s:=\left\{u\in L^1([s]):u_b\geq0,\quad
    \sum_{b=1}^s u_b=1\right\},
\end{equation}
% \begin{equation}
%     \Delta_s:=\left\{u: [s]\to [0,1], \, b\mapsto u_b:
%     \sum_{b=1}^s u_b=1\right\},
% \end{equation}
the set of probability distributions on the finite or countable
outcome set $[s]$, where $[\infty]:=\mathbb N$.
Alice receives an input $x\in X$, and Bob receives a target map
$f$ from a prescribed nonempty family of functions
\begin{equation}
    \cF\subseteq\{f:X\to\Delta_s\}.
    \label{eq:target-map-family}
\end{equation}
The desired output distribution is $f(x)$. Alice knows $x$ but not
$f$, while Bob knows $f$ but not $x$.

In the quantum setting, $X=\Pure$ and $s=d^2$. Each $s$-outcome POVM
$\mathsf E=\{E_b\}_{b=1}^s$ defines a target map through Born's rule:
\begin{equation}
    f_{\mathsf E,b}(\psi):=\operatorname{tr}(E_b\rho_\psi).
\end{equation}
The corresponding family is
$\cF_{d,s}:=\{f_{\mathsf E}:\mathsf E\in\mathfrak P_{d,s}\}$, where $\mathfrak P_{d,s}$ denotes the set of $d$-dimensional and $s$-outcome POVMs.

Following the quantum setting, we define simulation protocols
for an arbitrary target family $\cF$ as follows.
\begin{definition}[Exact simulation of a family of maps]
A strategy consists of shared randomness $\lambda$ on a probability
space $(\Omega,\mu)$, independent of the inputs, an encoder
$M_\lambda:X\to\Alp{\lambda}$ with a finite alphabet size $\AlpS_\lambda$, and a decoder
$q(\cdot\mid f,\lambda,m)\in\Delta_s$. The strategy exactly simulates $\cF$ if
\begin{equation}
    f_b(x)=\int_\Omega
    q\bigl(b\mid f,\lambda,M_\lambda(x)\bigr)
    \,\mathrm d\mu(\lambda)
    \label{eq:general-map-exactness}
\end{equation}
for every $x\in X$, $f\in\cF$, and $b\in[s]$. Communication costs
are defined as in 
Eqs.~\eqref{eq:worstcase-cost}, {\eqref{eq:expected-alphabet-size}} and \eqref{eq:expected-bit-cost}.
\end{definition}

\subsection{Subsets, envelopes, and mass deficits}

For a nonempty subset $R\subseteq X$ and $f\in\cF$, the lower
envelope of $R$ is
\begin{equation}
    \ell_b^f(R):=\inf_{x\in R}f_b(x).
    \label{eq:abstract-lower-envelope}
\end{equation}
Since $0\leq\ell_b^f(R)\leq f_b(x)$ for every $x\in R$, it is a
subnormalized distribution:
\begin{equation}
    \ell^f(R)\in\Delta_s^\circ
    :=\{u\in L^1([s]):u_b\geq0,\ |u|\leq1\},
    \label{eq:subnormalized-simplex}
\end{equation}
where $|u|:=\sum_bu_b$. 
A natural candidate for Bob's decoder on a cell $R$ is
$q_b=\ell_b^f(R)/|\ell^f(R)|$, whenever $|\ell^f(R)|>0$.
This choice never outputs an outcome whose target probability
vanishes at any input in $R$. The lower envelope therefore
identifies probability mass that can safely be assigned to every
input in the cell.

Define the mass deficit of $R$ with
respect to the target family by
\begin{equation}
    \kappa_{\cF}(R):=
    \sup_{f\in\cF}\bigl(1-|\ell^f(R)|\bigr)
    =1-\inf_{f\in\cF}\sum_{b=1}^s\inf_{x\in R}f_b(x).
    \label{eq:family-mass-deficit}
\end{equation}
This measures the largest deficit in the probability mass common
to all inputs in $R$. It satisfies $0\leq\kappa_{\cF}(R)\leq1$,
also when $s=\infty$. For every $f\in\cF$,
\begin{equation}
    |\ell^f(R)|\geq1-\kappa_{\cF}(R).
    \label{eq:general-overlap-bound}
\end{equation}
We write $u\preceq v$ for coordinatewise inequality. In particular,
$0\preceq\ell^f(R)\preceq f(x)$ for every $x\in R$.
If $R\subseteq P$, then $\ell^f(P)\preceq\ell^f(R)$ and
$\kappa_{\cF}(R)\leq\kappa_{\cF}(P)$.

\begin{definition}[Mass-deficit Minkowski dimension]
\label{def:mass-deficit-minkowski}
For $\epsilon>0$, let $N_{\cF}(\epsilon)$ be the smallest number of subsets of $X$, each of mass deficit at most $\epsilon$ with
respect to $\cF$, whose union is $X$. Set
$N_{\cF}(\epsilon)=\infty$ if no finite cover exists. The
\emph{mass-deficit Minkowski dimension} of $(X,\cF)$ is
\begin{equation}
    d_{\mathrm{mdM}}(X,\cF):=
    \limsup_{\epsilon\downarrow0}
    \frac{\log N_{\cF}(\epsilon)}{\log(1/\epsilon)}.
    \label{eq:mass-deficit-minkowski-dimension}
\end{equation}
\end{definition}
The definition follows the covering formulation of Minkowski,
or box-counting, dimension \cite[Sec.~3.1]{Falconer2003}
\cite[Ch.~5]{Mattila1995}, with the covering scale determined
by the mass deficit in \cref{eq:family-mass-deficit}.

Somewhat surprisingly, we prove that if $(X,\cF)$ has finite mass-deficit Minkowski dimension then exact simulation with finite expected (one way) communication cost is possible. More precisely, our general simulation theorem reads:

\begin{theorem}[Exact simulation with finite expected communication]
\label{thm:general-simulation}
If $(X,\cF)$ has finite mass-deficit Minkowski dimension, then $\cF$
admits a one-way exact simulation with finite expected communication. Furthermore, there exist constants $C_{\mathrm{tail}},\beta>0$ such that 
%$\Prob(\ell(\lambda)\geq k)\leq C_{\mathrm{tail}}2^{-\beta k}$ 
 $\Prob(\log \AlpS_\lambda \geq k)\leq C_{\mathrm{tail}}\, 2^{-\beta k}$
for every integer $k\geq0$.
\end{theorem}

Before presenting the proof we devote a subsection to explaining how the main ingredients of the proof naturally arise. Once these preliminary considerations are settled, the proof becomes elegant and short.

\subsection{Partitions and nested sequence of partitions}

A finite partition $\cP=\{R_m\}_{m=1}^N$ of $X$ consists of
nonempty, pairwise disjoint subsets whose union is $X$; in formulas,
\begin{equation}
    X=\bigsqcup_{m=1}^N R_m.
\end{equation}
For a fixed value of the shared random variable $\lambda$, the encoder can equivalently be viewed as a labelled partition of $X$. Indeed, one can set
\begin{equation}
    \cC_\lambda:=\{M_\lambda^{-1}(m):m\in\Alp{\lambda}\}
    \setminus\{\emptyset\}.
    \label{def:Cell}
\end{equation}
Bob receives message $m$, which he can equivalently interpret as the information $x\in R_m$, where $R_m:=M_\lambda^{-1}(m)$. To measure how fine a partition is, we use the mass deficits of its cells: for a finite partition, set
\begin{equation}
    \kappa_{\cF}(\cP):=\max_{R\in\cP}\kappa_{\cF}(R).
    \label{eq:partition-mass-deficit}
\end{equation}

Recall that in our setting the shared random variable also determines the amount of information that Alice can send to Bob. It is natural to consider encodings where each additional bit transmitted by Alice refines the information already conveyed. In partition language, this corresponds to requiring the cells of a finer partition (used when $\AlpS_\lambda$ is larger) to be contained in cells of a coarser partition. Thus we introduce the notion of a nested sequence of partitions.

\begin{definition}[Nested partitions]
\label{def:nested-partitions}
A sequence $\{\cC_r\}_{r\geq0}$ of finite partitions of $X$ is
\emph{nested} if $\cC_0=\{X\}$ and every cell of $\cC_{r+1}$
is contained in a cell of $\cC_r$.
\end{definition}
The unique cell $P\in\cC_r$ containing $R\in\cC_{r+1}$ is its
\emph{parent}, and $R$ is a \emph{child} of $P$.

As an example of a simple encoding strategy that naturally leads to a nested sequence of partitions, consider the following: let $X=[0,1)$ be the input space, and let the shared random variable $\lambda$ take values in the positive integers and set $\AlpS_\lambda:=2^\lambda$. In other words, the role of $\lambda$ is just to tell both parties how many bits Alice will send to Bob.
The simplest encoding that Alice can use to send Bob good approximations of the input is to send Bob the first $\lambda$ bits of the binary expansion of the input $x$. This corresponds to the nested sequence of partitions
\begin{equation*}
    \cC_r:=\left\{
        \left[\frac{m-1}{2^r},\frac{m}{2^r}\right):
        m\in[2^r]
    \right\},\qquad r\geq0.
\end{equation*}

In the quantum case, suitable nested partitions can have increasingly
small cells in trace distance, while retaining a cardinality polynomial in the inverse trace distance size. So Alice can specify the cell containing the inputs to Bob using \lq\lq a few" bits of communication, while Bob gets to know which probability distribution he is supposed to reproduce with \lq\lq good" approximation. For general $X$ and $\cF$, such
partitions with controlled cardinality need not exist, thus we introduce the following definition to rigorously generalize this property of the quantum setting.  
\begin{definition}[Well-partitionability]
\label{def:well-partitionable}
A set $X$ is \emph{well-partitionable with respect to $\cF$} if
there are nested partitions $\{\cC_r\}_{r\geq0}$ and constants
$a,\nu,C>0$ such that, for every integer $r\geq0$,
\begin{equation}
    |\cC_r|\leq2^{a+\nu r},
    \qquad\kappa_{\cF}(\cC_r)\leq C2^{-r}.
    \label{eq:partitionability-bounds}
\end{equation}
\end{definition}
The careful reader will have noticed that this property seems linked to finite mass-deficit Minkowski dimension. Indeed, lemma \ref{lem:dimension-partitionability}, proved in
Appendix~\ref{app:dimension-partitionability}, shows that whenever the mass-deficit Minkowski dimension of a pair $(X,\cF)$ is finite, $X$ is well-partitionable with respect to $\cF$, and vice versa. In particular, the construction gives nested partitions with any
$\nu>4d_{\mathrm{mdM}}(X,\cF)$.

For the quantum family $\cF_{d,d^2}$, the partitions in
Appendix~\ref{app:quantum-partitions} satisfy\footnote{The relation between the mass-deficit of a partition used here and the trace-distance diameter used in the appendix is quantified in lemma \ref{lem:Born-stability} }
\begin{equation}
    \begin{aligned}
        |\cC_r|&\leq d\,2^{(2d-2)r},\\
        \kappa_{\cF_{d,d^2}}(\cC_r)
        &\leq4d\sqrt{2d-2}\,2^{-r}.
    \end{aligned}
    \label{eq:quantum-partition-example}
\end{equation}
For families with finitely many outcomes, coordinate refinements
give similar bounds on compact smooth manifolds when the target
maps have uniformly bounded coordinate derivatives.

\subsection{Proof of the general simulation theorem}
\label{sec:multiscale}

\begin{proof}
By lemma \ref{lem:dimension-partitionability}, there are nested partitions
$\{\cC_r\}_{r\geq0}$ satisfying
\cref{eq:partitionability-bounds} for some $a,\nu,C>0$. Set
\begin{equation}
    \begin{aligned}
        N_0&:=\left\lceil\log\max\{1,C\}\right\rceil,\\
        \cP_k&:=\cC_{N_0+k}\quad(k\geq1),
    \end{aligned}
    \label{eq:general-initial-resolution}
\end{equation}
and take $\cP_0=\{X\}$ as the root. For $R\in\cP_k$,
\begin{equation}
    \kappa_{\cF}(R)\leq C2^{-N_0-k}\leq2^{-k}.
    \label{eq:general-mass-deficit-tail-bound}
\end{equation}
The protocol has the following three ingredients.

\emph{Shared randomness.}
Alice and Bob share a positive integer $\lambda$, independent of their
inputs, with
\begin{equation}
    \Prob(\lambda=k)=w_k:=2^{-k},\qquad k\geq1.
    \label{eq:general-resolution-protocol}
\end{equation}
Thus both parties know that they are using partition $\cP_\lambda$.
Write
\begin{equation}
    \begin{aligned}
        A_k&:=\sum_{j=1}^k w_j=1-2^{-k},
        \end{aligned}
    \label{eq:geometric-tail}
\end{equation}
with $A_0=0$. Before the protocol starts they have fixed a labelling of every partition $\cP_k$.

\emph{Alice's encoder.}
 Let $R_k(x)$ be
its cell containing $x$. When $\lambda=k$, Alice sends the label
$M_k(x)$ of $R_k(x)$. The alphabet has
$ \AlpS_k:=|\cP_k|$ symbols, so its label requires
\begin{equation}
   % \ell(k):=\log  |\cX_k|
   \log \AlpS_k
    \leq a+\nu(N_0+k)
    \label{eq:general-message-length}
\end{equation}
bits. Since $\lambda$ is shared, she need not transmit it.

\emph{Bob's decoder.}
Knowing $f$, $\lambda=k$, and the received label, Bob knows the cell
$R\in\cP_k$ and its chain of ancestors. We construct a probability
vector $q_{f,k}(R)$ from this information.
For each fixed $f$, the construction proceeds by induction on the
level $k$, together with cumulative sub-distributions $S_{f,k}(R)$.
The induction maintains
$0\preceq S_{f,k}(R)\preceq\ell^f(R)$ and
$|S_{f,k}(R)|=A_k$ for every cell $R\in\cP_k$.
The base case is the root at level zero. At each subsequent level,
we use the sub-distributions already assigned to all parent cells
to define the decoder and cumulative sub-distribution for every
child, preserving these two conditions.

By \cref{eq:general-overlap-bound,eq:general-mass-deficit-tail-bound},
\begin{equation}
    |\ell^f(R)|\geq1-2^{-k}=A_k.
    \label{eq:abstract-overlap-condition}
\end{equation}
Start with the cumulative sub-distribution $S_{f,0}(X)=0$.
Suppose that the parent $P\in\cP_{k-1}$ has already been assigned
$S_{f,k-1}(P)\preceq\ell^f(P)$ of mass $A_{k-1}$. Since
$\ell^f(P)\preceq\ell^f(R)$, the residual vector
\begin{equation}
    v_{f,k}(R):=\ell^f(R)-S_{f,k-1}(P)
\end{equation}
is nonnegative, and
\begin{equation}
    |v_{f,k}(R)|
    =|\ell^f(R)|-A_{k-1}\geq A_k-A_{k-1}=w_k>0.
    \label{eq:available-residual-mass}
\end{equation}
Define
\begin{align}
    q_{f,k}(R)&:=\frac{v_{f,k}(R)}{|v_{f,k}(R)|},
    \label{eq:recursive-decoder}\\
    S_{f,k}(R)&:=S_{f,k-1}(P)+w_kq_{f,k}(R).
    \label{eq:S-recursion}
\end{align}
Then $q_{f,k}(R)\in\Delta_s$, and
$w_k/|v_{f,k}(R)|\leq1$ ensures
\begin{equation}
    S_{f,k}(R)\preceq\ell^f(R),
    \qquad |S_{f,k}(R)|=A_k.
    \label{eq:cumulative-invariants}
\end{equation}
This completes the recursion. Bob outputs $b$ with probability
$q_{f,k,b}(R)$. The decoder depends only on his target map, the
shared level, and the received label.

For a fixed input $x$, the recursion telescopes along its ancestors:
\begin{equation}
    \sum_{j=1}^k w_jq_{f,j}(R_j(x))=S_{f,k}(R_k(x)).
    \label{eq:cumulative-decoder}
\end{equation}
By \cref{eq:cumulative-invariants}, the difference between $f(x)$
and this partial sum is nonnegative and has total mass $1-A_{k}$.
Consequently,
\begin{equation}
    \left\|f(x)-\sum_{j=1}^k w_jq_{f,j}(R_j(x))\right\|_1
    =2^{-k}\longrightarrow0.
    \label{eq:S-telescoping}
\end{equation}
Taking the limit gives
\begin{equation}
    f(x)=\sum_{k=1}^\infty2^{-k}q_{f,k}(R_k(x)),
    \label{eq:abstract-simulation-identity}
\end{equation}
which is exactly \cref{eq:general-map-exactness}. The argument is
valid for $f$ taking values in  $L^1([s])$, including the case of $s=\infty$, i.e,  countably infinite outcomes.

Finally, $\Exp[\lambda]=2$ and \cref{eq:general-message-length} give
\begin{equation}
    \begin{aligned}
        \Exp[\log \AlpS_\lambda]&=\sum_{k=1}^\infty2^{-k}\log  \AlpS_k \\
        &\leq a+\nu(N_0+2)<\infty.
    \end{aligned}
    \label{eq:general-average-bits}
\end{equation}
Since $\Prob(\lambda>t)\leq2^{1-t}$ for every real $t$,
\cref{eq:general-message-length} gives, for every integer $k\geq0$,
\begin{equation}
    \begin{aligned}
    \Prob\!\left(
    %\ell(\lambda)
     \log \AlpS_\lambda
    \geq k\right)
    &\leq\Prob\!\left(\lambda>\frac{k-1-a}{\nu}-N_0\right)\\
    &\leq2^{1+N_0+(a+1)/\nu}\,2^{-k/\nu}.
    \end{aligned}
    \label{eq:general-bit-tail}
\end{equation}
Thus the thesis holds with
$C_{\mathrm{tail}}=2^{1+N_0+(a+1)/\nu}$ and $\beta=1/\nu$.
\end{proof}

The overlap condition explains why coarse partitions are compatible
with exactness. At level $k$, the cumulative decoder assigns mass
$A_k$ valid for every input in the cell; the remaining mass $1-A_k$
is reserved for finer levels. For this fixed nested encoder and
level distribution, \cref{eq:abstract-overlap-condition} is also
necessary: any exact decoder has a cumulative subdistribution of
mass $A_k$ below $f(x)$ for every $x\in R$, and hence below
$\ell^f(R)$.

\subsection{Quantum simulation as a corollary}
\label{sec:nested-colorings}

For the case of simulating qudit correlations, it is easy to construct \lq\lq good" nested partitions with respect to trace distance; one explicit example is given in Appendix \ref{app:quantum-partitions}. We show how the trace distances relates to the deficit mass and thus that the nested partitions aforementioned can be used for the exact simulation protocol.  As explained in section \ref{sec:general-framework}, simulating POVMs with $d^{2}$ outcomes is sufficient for the simulation of arbitrary POVM.
To bound the mass deficit, use the pure-state trace distance
\begin{equation}
    D(\psi,\phi):=\frac12\|\rho_\psi-\rho_\phi\|_1
    =\sqrt{1-|\braket{\psi}{\phi}|^2}.
    \label{eq:pure-trace-distance}
\end{equation}
Write $\diam_D(R):=\sup_{\psi,\phi\in R}D(\psi,\phi)$ and
$\diam_D(\cP):=\max_{R\in\cP}\diam_D(R)$.

\begin{lemma}[Mass deficit of Born probabilities]
\label{lem:Born-stability}
For every nonempty $R\subseteq\Pure$,
\begin{equation}
    \kappa_{\cF_{d,s}}(R)\leq d\,\diam_D(R).
    \label{eq:Born-mass-deficit}
\end{equation}
Moreover, every POVM $\mathsf E\in\mathfrak P_{d,s}$ satisfies
\begin{equation}
    \sum_{b=1}^s\inf_{\psi\in R}p_b(\mathsf E,\psi)
    \geq1-d\diam_D(R).
    \label{eq:multiway-overlap}
\end{equation}
\end{lemma}

\begin{proof}
The variational characterization of trace distance gives
\begin{equation}
    |p_b(\mathsf E,\psi)-p_b(\mathsf E,\phi)|
    \leq\|E_b\|_\infty D(\psi,\phi).
    \label{eq:effect-variation}
\end{equation}
Positivity and POVM normalization
give
\begin{equation}
    \sum_{b=1}^s\|E_b\|_\infty
    \leq\sum_{b=1}^s\operatorname{tr}E_b=d.
\end{equation}
Fix $\psi_0\in R$. By \cref{eq:effect-variation},
\begin{equation}
    \begin{aligned}
    &1-\sum_{b=1}^s\inf_{\phi\in R}p_b(\mathsf E,\phi)\\
    &\quad=\sum_{b=1}^s\sup_{\phi\in R}
      \bigl(p_b(\mathsf E,\psi_0)-p_b(\mathsf E,\phi)\bigr)\\
    &\quad\leq\sum_{b=1}^s\|E_b\|_\infty\diam_D(R)
    \leq d\,\diam_D(R).
    \end{aligned}
\end{equation}
Taking the supremum over $\mathsf E$ proves
\cref{eq:Born-mass-deficit}, and the same estimate gives
\cref{eq:multiway-overlap}.
\end{proof}

Appendix~\ref{app:quantum-partitions} supplies nested partitions with
$|\cC_r|\leq d2^{n_{d}r}$ and
$\diam_D(\cC_r)\leq4\sqrt{n_d}\,2^{-r}$, where $n_d=2d-2$.
By \cref{eq:Born-mass-deficit}, they satisfy
\cref{eq:quantum-partition-example}. Thus $\Pure$ is
well-partitionable with respect to $\cF_{d,d^2}$, and
\cref{lem:dimension-partitionability} gives
\begin{equation}
    d_{\mathrm{mdM}}(\Pure,\cF_{d,d^2})\leq n_d.
    \label{eq:quantum-mass-deficit-minkowski-bound}
\end{equation}

\begin{proof}[Proof of \cref{thm:main-simulation}]
For $d\geq2$, the quantum family satisfies the hypothesis of
\cref{thm:general-simulation}. To obtain the explicit communication
bound, use the partitions above, with
\begin{equation}
    a=b_d:=\log d,\qquad
    \nu=n_d,\qquad C=dC_d,
\end{equation}
where $C_d=4\sqrt{n_d}$. The initial resolution is therefore
\begin{equation}
    N_0:=\lceil\log(4d\sqrt{n_d})\rceil.
    \label{eq:explicit-N0}
\end{equation}
Use the same shared level $\lambda$ with $\Prob(\lambda=j)=2^{-j}$ as in the
general theorem. At $\lambda=j$, Alice sends the label of her cell in
$\cC_{N_0+j}$, and Bob uses the decoder
$q_{f_{\mathsf E},j}$ from \cref{eq:recursive-decoder}. Thus
\begin{equation}
    p_b(\mathsf E,\psi)
    =\sum_{j=1}^\infty2^{-j}
    q_{f_{\mathsf E},j,b}(R_j(\psi)).
    \label{eq:quantum-exact-expansion}
\end{equation}
Alice can specify her cell by a chart label and $N_0+j$ binary
subdivision choices for each of its $n_d$ real coordinates, giving
\begin{equation}
    \ell(j)=\log|\cC_{N_0+j}|
    \leq b_d+n_d(N_0+j).
    \label{eq:quantum-message-length}
\end{equation}
Since $\Exp[\lambda]=2$,
\begin{equation}
    \Exp[\log \AlpS_\lambda]\leq b_d+n_d(N_0+2),
    \label{eq:quantum-average-bits}
\end{equation}
which is the stated bound. For every integer $k\geq0$,
\begin{equation}
    \begin{aligned}
    \Prob\!\left(
    %\ell(\lambda)
    \log \AlpS_\lambda>k\right)
    &\leq\Prob\!\left(\lambda>\frac{k-b_d}{n_d}-N_0\right)\\
    &\leq2^{1+N_0+b_d/n_d}\,2^{-k/n_d}.
    \end{aligned}
    \label{eq:quantum-bit-tail}
\end{equation}
Thus one may take $C_{d}=2^{1+N_0+b_d/n_d}$ and
$\beta=1/n_d=1/(2d-2)$ for this tail bound.
Mixed preparations are simulated by sampling a pure-state
decomposition and applying the same encoder; linearity preserves
the Born probabilities and the communication bound. Bob's local
POVM decomposition and relabelling then extend the protocol to
arbitrary outcome spaces.
\end{proof}

The average number of symbols, interpreted as the average classical
dimension, diverges. For $d\geq4$, this is required by
\cref{thm:no-finite-average-classical-dimension}; for this
construction it holds for every $d\geq2$.
Indeed, the image of each chart in
Appendix~\ref{app:quantum-partitions} contains the coordinate cube
$[-1/2,1/2]^{n_d}$. Hence, for $r\geq1$, at least
$2^{n_d(r-1)}$ cells per chart are nonempty, and
\begin{equation}
    \AlpS_j:=|\cC_{N_0+j}|
    \geq d\,2^{n_d(N_0+j-1)}.
    \label{eq:level-classical-dimension}
\end{equation}
It follows that
\begin{equation}
    \Exp[\AlpS_\lambda]\geq
    d\,2^{n_dN_0-1}
    \sum_{j=0}^{\infty}2^{(n_d-1)j}=\infty.
    \label{eq:protocol-infinite-mean-alphabet}
\end{equation}
The series diverges because $n_d\geq2$.

\subsection{Approximate simulation with bounded alphabet size}

Characterizing the cost of exactly simulating finite-dimensional quantum communication is interesting from a fundamental perspective. From a more practical point of view, it is more relevant to consider the task of approximate
simulation. Indeed, for any finite number of rounds, sufficiently close
probability distributions cannot be distinguished reliably. Here we show that our simulation protocol for exact simulation, requiring finite average number of bits, can be adapted to a protocol for approximate simulation with finite alphabet size (in the worst-case).  

We start by defining the notion of approximate simulation.

\begin{definition}[$\epsilon$-accurate classical one-way simulation]
\label{def:Approximateclassicalone-waySimul}
The parties share a random variable $\lambda$ on a probability
space $(\Omega,\mu)$, independent of their inputs. An encoder is a map
\begin{equation}
    M_\lambda:\Pure\longrightarrow \Alp{\lambda},
\end{equation}
where $\AlpS_{\lambda}\in\mathbb N$ is positive. A decoder assigns a probability
distribution $q(\cdot\mid\mathsf E,\lambda,m)$ on the measurement
outcomes, as in definition \ref{def:classicalone-waySimul}.
Let
\begin{equation}
    [q(\mathsf E,\psi)]_b:=
    \int_\Omega
    q\bigl(b\mid\mathsf E,\lambda,M_\lambda(\psi)\bigr)
    \,\mathrm d\mu(\lambda)
\end{equation}
be the average probability distribution that Bob outputs.
The simulation is $\epsilon$-accurate if, for every input state
$\psi$ and measurement $\mathsf E$,
\begin{equation}
    D\!\left(
    q(\mathsf E,\psi)
    \,\middle\|\,
    p(\mathsf E,\psi)
    \right)\leq\epsilon.
    \label{eq:approx-simulation-general}
\end{equation}
\end{definition}

Here $D(\cdot\|\cdot)$ denotes the Kullback--Leibler divergence
\cite{KullbackLeibler1951}, measured in bits:
\begin{equation}
    D(q\|p):=\sum_b q_b\log\frac{q_b}{p_b}.
\end{equation}
Terms with $q_b=0$ contribute zero, while $q_b>0=p_b$ makes the
divergence infinite. Other distinguishability measures could be
adopted. In contrast with previous works \cite{Montina2011Approximate,Montina2013Exponential,ZartabEtAl2026} we use relative entropy because of its operational
characterization through the Chernoff--Stein lemma: it determines
the optimal error exponent in asymmetric hypothesis testing between
independent samples from the two distributions
\cite[Sec.~11.8]{CoverThomas2006}. Notice that closeness in relative entropy is a stronger condition than, e.g., total variation, as the latter can be arbitrarily small despite a diverging relative entropy. 

Consider the following protocol, which requires no shared randomness.
Take a sufficiently fine level $r$ of the quantum nested partitions
constructed in Appendix~\ref{app:quantum-partitions}, and write
\begin{equation}
    \cC_r=\{R_1,\ldots,R_N\},
    \qquad N=|\cC_r|.
\end{equation}
Alice sends the label
\begin{equation}
    M(\psi):=m
    \qquad\text{if }\psi\in R_m.
\end{equation}
Receiving $m$ therefore tells Bob that the input state belongs to
$R_m$.

Bob uses the normalized lower-envelope decoder. Writing
\begin{equation}
    \ell_b(\mathsf E,R):=
    \inf_{\phi\in R}p_b(\mathsf E,\phi),
    \qquad
    \abs{\ell(\mathsf E,R)}:=\sum_b\ell_b(\mathsf E,R),
\end{equation}
he outputs $b$ with probability
\begin{equation}
    q(b\mid\mathsf E,m):=
    \frac{\ell_b(\mathsf E,R_m)}
         {\abs{\ell(\mathsf E,R_m)}}.
    \label{eq:approx-lower-envelope-decoder}
\end{equation}
This is the same decoder used at the first level of the exact
simulation protocol. The communication cost is at most $\log N$
bits.

\begin{proposition}[$\epsilon$-accurate simulation with finite alphabet-size]
\label{prop:approx-lower-envelope} For sufficiently large $r$, the above protocol using $\cC_{r}$ partitions for the encoding (see Appendix \ref{app:quantum-partitions}), and the normalized lower envelope for the decoding, uses an alphabet-size $N=\abs{\cC_{r}}$ and achieves $\epsilon-$accurate simulation, for 
\begin{equation}
    \epsilon \geq \frac{K_d}{N^{\alpha_d}},
    \qquad
    \alpha_d=\frac1{d-1},
    \label{eq:approx-relative-entropy-bound}
\end{equation}
 where $K_d$ depends only on $d$.
\end{proposition}
\begin{proof}
    The proof follows after straightforward calculations based on the mass-deficit size of the cells, see Appendix \ref{app:ApproxSimulation} for details. 
\end{proof}

This result shows that, for every fixed $d\geq4$,
$\epsilon$-accurate simulation of a qudit is possible using an
alphabet of size of order $(K_d/\epsilon)^{d-1}$, corresponding to
$(d-1)\log(K_d/\epsilon)+O_d(1)$ bits of communication.
Remarkably, the protocol requires no shared randomness.

A natural candidate for improving its accuracy is to exploit shared
randomness by Haar-randomly rotating the chosen partition.
More precisely, let the shared randomness consist of a
Haar-random unitary $U\in U(d)$, and let Alice's encoding be
\begin{equation}
    M_U(\psi):=M(U\psi),
\end{equation}
where $M(\phi)$ labels the cell containing $\phi$.
Since Bob knows $U$, receiving the label $m$ tells him that
$\psi\in U^{-1}R_m$, and he applies the same normalized
lower-envelope decoder to this rotated cell.
Intuitively, averaging over orientations removes the dependence
on the absolute orientation of the partition and may allow
deviations of the decoded probabilities from the Born probabilities
to cancel. By convexity of relative entropy, this randomization
cannot increase the worst-case error, so
Eq.~\eqref{eq:approx-relative-entropy-bound} remains valid.
We leave as an open question whether these cancellations improve the asymptotic dependence on $N$, beyond a possible improvement of the constant.

\section{Simulation of Gaussian CV quantum systems}
\label{sec:gaussian}

Finally, to illustrate the versatility of the general simulation theorem we consider the simulation of a continuous variable (CV) quantum system. For simplicity we present the single-mode case. The generalization to $n$ modes is discussed at the end of this section and proven in Appendix~\ref{app:multimode-gaussian}.

\subsection{Single-mode Gaussian systems}

In the context of CV quantum systems, pure single-mode Gaussian states are displaced squeezed vacua
\begin{equation}
    \ket{\alpha,\xi}:=e^{\alpha a^\dagger-\alpha^*a}\,
    e^{\frac12(\xi^*a^2-\xi a^{\dagger2})}\ket0,
    \label{eq:gaussian-pure-states}
\end{equation}
with displacement and squeezing  parameters $\alpha,\xi \in\mathbb C$. The pure single-mode Gaussian measurements are the so-called general-dyne measurements~\cite{Serafini2017,Weedbrook2012}, continuous-outcome POVMs
\begin{equation}
    E_{\beta|\zeta}:=\frac1\pi\ketbra{\beta,\zeta},\qquad \beta\in\mathbb C,
    \label{eq:gaussian-generaldyne}
\end{equation}
specified by the measurement-squeezing parameter $\zeta$. In particular, the POVM $\{E_{\beta|\zeta}\}$ is a heterodyne measurement for $\zeta=0$; while in the limit $\zeta=re^{2\ii\varphi}$ with $r\to\infty$, its measured-quadrature marginal converges to that of a homodyne measurement of the optical quadrature $\hat q _\varphi=\tfrac12(e^{\ii\varphi}a^\dagger+e^{-\ii\varphi}a)$.
The conjugate outcome is discarded, since its variance diverges in this limit.

By definition the Gaussian states have Gaussian Wigner functions~\cite{Weedbrook2012,AdessoRagyLee2014}, i.e.\ densities of normal distributions. They are therefore uniquely specified by the mean $\bm\mu = \binom{\langle{\hat x} \rangle}{\langle {\hat p}\rangle}$ and the (positive-semidefinite) covariance matrix $\Sigma$ of the complementary optical quadratures $\hat x=\tfrac12(a+a^\dagger)$ and $\hat p=\tfrac\ii2(a^\dagger-a)$,
see Appendix~\ref{app:multimode-gaussian} for more details. In particular, for the states $\ketbra{\alpha,\xi}$ in \cref{eq:gaussian-pure-states} one finds
\begin{align}
\bm\mu&=(\mathrm{Re}\,\alpha, \ \mathrm{Im}\,\alpha)^T \quad \text{and}\nonumber\\
    \Sigma&=\frac14\left(\begin{smallmatrix}
\cosh 2|\xi|-\frac{\sinh 2|\xi|}{|\xi|}\,\mathrm{Re}\,\xi & -\frac{\sinh 2|\xi|}{|\xi|}\,\mathrm{Im}\,\xi\\
-\frac{\sinh 2|\xi|}{|\xi|}\,\mathrm{Im}\,\xi & \cosh 2|\xi|+\frac{\sinh 2|\xi|}{|\xi|}\,\mathrm{Re}\,\xi
\end{smallmatrix}\right){\color{red}.}
    \label{eq:gaussian-moments}
\end{align}
Similarly, a general-dyne measurement in \cref{eq:gaussian-generaldyne} has outcome $\bm y=(\mathrm{Re}\,\beta,\mathrm{Im}\,\beta)^T$ and is specified by the covariance matrix $Z$ of its seed state $\ket{0,\zeta}$. Its output density on a Gaussian state $\rho\simeq(\bm\mu,\Sigma)$ is also Gaussian
\begin{align}
    f_{Z}(\bm y\mid\bm\mu,\Sigma)
    %&=\tr\!\big(\rho\,E_{\beta|\xi_M}\big) \nonumber\\
    &=\frac{e^{-\frac12(\bm y-\bm\mu)^T(\Sigma+Z)^{-1}(\bm y-\bm\mu)}}
    {2\pi\sqrt{\det(\Sigma+Z)}}.
    \label{eq: f dens}
\end{align}

We now consider a quantum communication scenario where Alice sends any Gaussian state to Bob, on which he performs any general-dyne measurement. Without further restriction the described communication channel has infinite classical capacity, and allows one to deterministically transmit any number of bits. Thus, to make the question interesting, we restrict the single-mode Gaussian states $\rho$ sent by Alice to have a bounded average energy  $\tr(\rho\,a^\dagger a)=\|\bm\mu\|^2+\tr\Sigma-\tfrac12\leq E$, for some $E>0$. Concretely, we consider the input set 
\begin{align}
    X_E:=\Big\{\rho\simeq(\bm\mu,\Sigma):\ &\|\bm\mu\|^2+\tr\Sigma-\tfrac12\leq E,\nonumber\\
    &\quad\quad \Sigma+\tfrac14\left(\begin{smallmatrix}0&\ii\\-\ii&0\end{smallmatrix}\right)\geq0\Big\},
    \label{eq:gaussian-input-set}
\end{align}
where the second condition, saturated by pure states\footnote{By saturation we mean that one eigenvalue of $\Sigma+\tfrac14\left(\begin{smallmatrix}0&\ii\\-\ii&0\end{smallmatrix}\right)$ is equal to zero.}, represents the uncertainty relations for complementary quadratures and characterizes the covariance matrices of quantum states. Notice that $X_E$ contains both pure and mixed Gaussian states.

The task of simulating this scenario with classical communication  presents two major differences with the simulation of qudits discussed in the last section. First, the CV states sent by Alice are not restricted to any finite-dimensional subspace. Second, general-dyne measurements are continuous-valued, and this continuity cannot be attributed to any classical mixedness. Despite these complications, Gaussian states with bounded average energy can also be simulated with a finite expected number of classical bits. Let us now sketch the main differences with the qudit simulation.

As the outcome $\bm y$ is continuous valued, the decoder must now simulate the continuous density \cref{eq: f dens}. For a subset $R\subset X_E$, we extend the mass deficit to continuous-output measurements via
\begin{align}
    \kappa_{\cF}(R)&:=
    1-\inf_{Z}\int_{\mathbb R^{2}}\mathrm d^{2}\bm y \inf_{(\bm\mu,\Sigma)\in R} f_{Z}(\bm y\mid\bm\mu,\Sigma),
    \label{eq:gaussian-family-mass-deficit}
\end{align}
which is well defined for the output densities in Eq.~\eqref{eq: f dens}.
Using the energy constraint, it is not difficult to see that covariance matrices in $X_E$ satisfy
\begin{align}
    &\gamma_E\one_2\leq\Sigma\leq\Gamma_E\one_2,\quad \text{with}
    \label{eq:gaussian-spectral-bounds}\\
    &\gamma_E:=\tfrac14\big(\sqrt{1+E}-\sqrt E\big)^2,\nonumber \quad
    \Gamma_E:=\tfrac14\big(\sqrt{1+E}+\sqrt E\big)^2,
\end{align}
see Proposition~\ref{prop:mm-SE-bounds}. 
As the covariance matrix is real symmetric $\Sigma =\Sigma^T$, these bounds (more precisely, the element-wise bounds they imply, see Proposition~\ref{prop:mm-SE-bounds}), together with $\|\bm\mu\|^2\leq E$, confine $X_E$ to an explicit box in the coordinates \begin{equation}
(\mu_1,\mu_2,\Sigma_{11},\Sigma_{12},\Sigma_{22})\in\mathbb R^5.
\end{equation}
In Appendix~\ref{app:multimode-gaussian} we show that the mass deficit of a cell of this box under the output densities in \cref{eq: f dens} is linear in its side length, uniformly in $Z$, and hence for all general-dyne measurements, including the homodyne limits. The energy bound $\Sigma+Z\geq \Sigma \geq \gamma_E \one_2$ is also crucial here, as it guarantees that the output densities  are not infinitely narrow.

To apply the general simulation theorem for continuous outcomes, \cref{thm:continuous-simulation}, it therefore suffices to slice the box into $2^{5r}$ cubes at nesting level $r$. This nested partition implies a simulation strategy with
\begin{equation}
    \Exp[\log \AlpS_\lambda]\leq5\left(\left\lceil\log\max\{1,K_{1,E}\}\right\rceil+2\right) = 10\log E+O(1)
    \label{eq:gaussian-average-bits-1mode}
\end{equation}
expected communicated bits, where $K_{1,E}$ is given by \cref{eq:gaussian-mass-deficit-constant} below. Indeed, the partitions have at most $2^{5r}$ cells and mass deficit at most $K_{1,E}2^{-r}$.

\subsection{Multimode generalization}

The same construction applies to $n$ modes. In this case the relevant states $X_E$ are all $n$-mode Gaussian states with mean photon number at most $E$, described by a $2n$ real displacement vector $\bm \mu$ and a $2n \times 2n$ real symmetric covariance matrix $\Sigma$.  Bob's measurements are all $n$-mode general-dyne measurements and $\cF$ denotes the family of their output densities. The parameter box now has dimension $D_n:=2n^2+3n$. Appendix~\ref{app:multimode-gaussian} shows that the set $X_E$ is well partitioned with respect to $\cF$ with $d_{\mathrm{mdM}}(X_E,\cF)\leq2n^2+3n$, and establishes the following result

\begin{theorem}
\label{prop:mm-simulation} The scenario with $n$-mode Gaussian states with bounded average photon number $E$ and general-dyne measurements admits a  
classical simulation with finite expected communication 
\begin{equation}
    \Exp[\log \AlpS_\lambda]\leq D_n\left(\left\lceil\log\max\{1,K_{n,E}\}\right\rceil+2\right),
    \label{eq:gaussian-average-bits}
\end{equation}
with
\begin{align}
    K_{n,E}:=&\left(\frac{\sqrt2\,n^{3/2}+4n^2}{\gamma_E}+\frac{4n}{\sqrt{\gamma_E}}+\frac14\right)\nonumber\\
    &\times\max\Big\{2\sqrt E,\sqrt{E(1+E)}\Big\},
    \label{eq:gaussian-mass-deficit-constant}
\end{align}
and an exponential
tail for the message length.
\end{theorem}

Since $K_{n,E}=O(n^2E^2)$, this gives $\Exp[\log \AlpS_\lambda]\leq D_n\big(2\log E+O(\log n)\big)$ in the limit of large energy. That is, the expected number of bits used by the classical simulation strategy grows logarithmically with the energy $E$ and quadratically with the mode number $n$, up to a logarithmic factor ($n^2\log n$).

We emphasize again that neither the dimension of the Hilbert space nor the cardinality of measurement outputs are bounded for CV systems. The present construction therefore cannot use the trace-distance shortcut we employed for qudits. We instead constructed a nested partition of the {\it parameter space} $\{(\bm\mu,\Sigma)\}$ with cells of increasingly small mass deficit with respect to all general-dyne measurements. The possibility to obtain such a construction with only a few extra steps nicely illustrates the advantages of the general simulation theorem.

\section{Discussion}

We have shown that the classical simulation of quantum correlations arising from qudits in prepare-and-measure scenarios requires a classical communication cost that is essentially as high as it could be. Specifically, we first showed that no simulation model involving classical messages of finite expected alphabet size can exactly reproduce qudit correlations for $d\geq4$. Second, we proved a general simulation theorem stating that the correlations of any physical theory (under some minimal requirements) can be classically simulated via a model using a finite expected number of classical bits. The case of finite-dimensional quantum systems is a special case of this theorem, requiring $O(d \log(d))$ bits on average.

Overall, these results identify the communication complexity of simulating the correlations from finite-dimensional quantum systems, both in the prepare-and-measure scenario and in bipartite Bell tests. On the more technical side, it would still be interesting to identify precisely the optimal scaling for the expected communication. Our protocol uses more than $k$-bits of communication with a probability  exponentially small in $k$, for a $d$-dimensional system we have $\Prob\!\left(\log \AlpS_\lambda >k\right)\leq C_{d}\,2^{-\beta_{d}k} $ where $\beta_{d}=(2d-2)^{-1}$. One interesting direction is to find the largest exponent $\beta^{*}_{d}$ for which the exact simulation is possible. Here we could show that $\beta_{d}\leq \beta_{d}^{*}\leq 1$ (from Theorem \ref{thm:no-finite-average-classical-dimension}) for $d\geq 4$. 

%it holds $\beta_{d}\leq \beta_{d}^{*}\leq 1$, where the bound is supplied by Theorem \ref{thm:no-finite-average-classical-dimension}.

An interesting question for future research is to see if our results can be effectively applied to the multipartite case (i.e. three parties or more), where the classical simulation costs of entanglement have also been investigated \cite{Buhrman2003,Branciard2011,Branciard2012,Brassard2016}. In particular, Brassard, Devroye and Gravel \cite{Brassard2018} provided a universal model requiring two-way communication between the parties, with $O(m^2)$ bits of communication on average, where $m$ denotes the number of parties. It would be interesting to see if our simulation model can provide more economical models, for example involving only one-way communication between parties.

Another direction is to apply our general simulation theorem to scenarios involving infinite-dimensional quantum systems, as we illustrated via bounded-energy Gaussian states subject to general Gaussian measurements. A first question is whether the latter scenario could in fact be simulated more effectively, using e.g. a model with finite expected alphabet size. Finding other classes of CV systems where the simulation applies would be interesting, as well as studying the implications of these results for tasks in quantum information processing.

Finally, our general simulation theorem goes well beyond quantum physics, and applies to a broad class of theories, including generalized probabilistic theories. It would be interesting to see whether there exist theories where correlations can be simulated more efficiently than in quantum theory, for example involving only a finite expected (or even worst-case) alphabet size.

%%%%%%%%%%%%%%%%%%%%%%%%%%%%%%
\section{Acknowledgements}
%%%%%%%%%%%%%%%%%%%%%%%%%%%%%%

The authors acknowledge financial support from the Swiss National Science Foundation (project 219366). This work was also funded with a grant from the Foundation for the University of Geneva (Funige).

%%%%%%%%%%%%%%%%%%%%%%%%%%%%%%
\section{AI statement}
%%%%%%%%%%%%%%%%%%%%%%%%%%%%%%
 ChatGPT Sol 5.6 and ChatGPT Astra 6.0 have contributed via extensive dialogues with the authors. This helped prove the statements, as well as identifying and correcting errors. All proofs, claims, and conclusions have been reviewed and verified by the authors who take full responsibility for the content.

\bibliography{biblio}

@article{Perry2015,
   title={Communication Tasks with Infinite Quantum-Classical Separation},
   volume={115},
   ISSN={1079-7114},
   url={http://dx.doi.org/10.1103/PhysRevLett.115.030504},
   DOI={10.1103/physrevlett.115.030504},
   number={3},
   journal={Physical Review Letters},
   publisher={American Physical Society (APS)},
   author={Perry, Christopher and Jain, Rahul and Oppenheim, Jonathan},
   year={2015},
   month=July }

@book{HornJohnson2012,
  author    = {Horn, Roger A. and Johnson, Charles R.},
  title     = {Matrix Analysis},
  edition   = {2},
  publisher = {Cambridge University Press},
  address   = {Cambridge},
  year      = {2012},
  isbn      = {978-0-521-83940-2}
}

@article{SimonMukundaDutta1994,
  author  = {Simon, R. and Mukunda, N. and Dutta, Biswadeb},
  title   = {Quantum-noise matrix for multimode systems: {U(n)} invariance, squeezing, and normal forms},
  journal = {Physical Review A},
  volume  = {49},
  pages   = {1567--1583},
  year    = {1994},
  doi     = {10.1103/PhysRevA.49.1567}
}

@misc{gavinsky2008,
      title={Exponential separations for one-way quantum communication complexity, with applications to cryptography}, 
      author={Dmytro Gavinsky and Julia Kempe and Iordanis Kerenidis and Ran Raz and Ronald de Wolf},
      year={2008},
      eprint={quant-ph/0611209},
      archivePrefix={arXiv},
      primaryClass={quant-ph},
      url={https://arxiv.org/abs/quant-ph/0611209}, 
}

@book{VogelWelsch2006,
  author    = {Vogel, Werner and Welsch, Dirk-Gunnar},
  title     = {Quantum Optics},
  edition   = {3},
  publisher = {Wiley-VCH},
  address   = {Weinheim},
  year      = {2006},
  isbn      = {978-3-527-40507-7}
}

@book{Serafini2017,
  author = {Serafini, Alessio},
  title = {Quantum Continuous Variables: A Primer of Theoretical Methods},
  publisher = {CRC Press},
  address = {Boca Raton},
  year = {2017},
  isbn = {9781482246346}
}

@article{AdessoRagyLee2014,
  author = {Adesso, Gerardo and Ragy, Sammy and Lee, Antony R.},
  title = {Continuous Variable Quantum Information: {G}aussian States and Beyond},
  journal = {Open Systems \& Information Dynamics},
  volume = {21},
  pages = {1440001},
  year = {2014},
  doi = {10.1142/S1230161214400010},
  eprint = {1401.4679},
  archivePrefix = {arXiv}
}

@article{ValcarceEtAl2022,
  author = {Valcarce, Xavier and Sekatski, Pavel and Gouzien, Elie and Melnikov, Alexey and Sangouard, Nicolas},
  title = {Automated design of quantum optical experiments for device-independent quantum key distribution},
  journal = {arXiv:2209.06468},
  year = {2022},
  eprint = {2209.06468},
  archivePrefix = {arXiv}
}

@article{Havlicek2020,
  title = {Simple communication complexity separation from quantum state antidistinguishability},
  author = {Havl{\'{\i}}{\ifmmode \check{c}\else \v{c}\fi{}}ek, Vojt{\ifmmode \check{e}\else \v{e}\fi{}}ch and Barrett, Jonathan},
  journal = {Phys. Rev. Res.},
  volume = {2},
  issue = {1},
  pages = {013326},
  numpages = {8},
  year = {2020},
  month = {Mar},
  publisher = {American Physical Society},
  doi = {10.1103/PhysRevResearch.2.013326},
  url = {https://link.aps.org/doi/10.1103/PhysRevResearch.2.013326}
}

@misc{brask2026,
      title={Quantum correlations in prepare-and-measure scenarios and their semi-device-independent applications}, 
      author={Jonatan Bohr Brask and Nicolas Brunner and Jef Pauwels and Davide Rusca and Armin Tavakoli},
      year={2026},
      eprint={2603.23604},
      archivePrefix={arXiv},
      primaryClass={quant-ph},
      url={https://arxiv.org/abs/2603.23604}, 
}

@article{Gallego2010,
  title = {Device-Independent Tests of Classical and Quantum Dimensions},
  author = {Gallego, Rodrigo and Brunner, Nicolas and Hadley, Christopher and Ac{\'{\i}}n, Antonio},
  journal = {Phys. Rev. Lett.},
  volume = {105},
  issue = {23},
  pages = {230501},
  numpages = {4},
  year = {2010},
  month = {Nov},
  publisher = {American Physical Society},
  doi = {10.1103/PhysRevLett.105.230501},
  url = {https://link.aps.org/doi/10.1103/PhysRevLett.105.230501}
}

@misc{schlosser2026,
      title={Bounding the classical cost of simulating quantum behaviors in the prepare-and-measure scenario}, 
      author={Sebastian Schlösser and Matthias Kleinmann},
      year={2026},
      eprint={2603.01255},
      archivePrefix={arXiv},
      primaryClass={quant-ph},
      url={https://arxiv.org/abs/2603.01255}, 
}

@article{Montanaro2019,
  doi = {10.22331/q-2019-06-28-154},
  url = {https://doi.org/10.22331/q-2019-06-28-154},
  title = {Quantum states cannot be transmitted efficiently classically},
  author = {Montanaro, Ashley},
  journal = {{Quantum}},
  issn = {2521-327X},
  publisher = {{Verein zur F{\"{o}}rderung des Open Access Publizierens in den Quantenwissenschaften}},
  volume = {3},
  pages = {154},
  month = jun,
  year = {2019}
}

@article{Montina2016,
doi = {10.1088/1367-2630/18/1/013035},
url = {https://doi.org/10.1088/1367-2630/18/1/013035},
year = {2016},
month = {jan},
publisher = {IOP Publishing},
volume = {18},
number = {1},
pages = {013035},
author = {Montina, Alberto and Wolf, Stefan},
title = {Information-based measure of nonlocality},
journal = {New Journal of Physics}
}

@article{Renner2023,
  doi = {10.22331/q-2023-10-24-1149},
  url = {https://doi.org/10.22331/q-2023-10-24-1149},
  title = {The minimal communication cost for simulating entangled qubits},
  author = {Renner, Martin J. and Quintino, Marco T{\'{u}}lio},
  journal = {{Quantum}},
  issn = {2521-327X},
  publisher = {{Verein zur F{\"{o}}rderung des Open Access Publizierens in den Quantenwissenschaften}},
  volume = {7},
  pages = {1149},
  month = oct,
  year = {2023}
}

@article{DArianoLoPrestiPerinotti2005,
  author = {D'Ariano, Giacomo Mauro and Lo Presti, Paoloplacido and Perinotti, Paolo},
  title = {Classical randomness in quantum measurements},
  journal = {Journal of Physics A: Mathematical and General},
  volume = {38},
  number = {26},
  pages = {5979--5991},
  year = {2005},
  doi = {10.1088/0305-4470/38/26/010},
  eprint = {quant-ph/0408115},
  archivePrefix = {arXiv}
}

@article{ChiribellaDArianoSchlingemann2007,
  author = {Chiribella, Giulio and D'Ariano, Giacomo Mauro and Schlingemann, Dirk-Michael},
  title = {How Continuous Quantum Measurements in Finite Dimension Are Actually Discrete},
  journal = {Physical Review Letters},
  volume = {98},
  pages = {190403},
  year = {2007},
  doi = {10.1103/PhysRevLett.98.190403},
  eprint = {quant-ph/0702068},
  archivePrefix = {arXiv}
}

@book{Falconer2003,
  author = {Falconer, Kenneth},
  title = {Fractal Geometry: Mathematical Foundations and Applications},
  edition = {2},
  publisher = {John Wiley \& Sons},
  year = {2003},
  doi = {10.1002/0470013850}
}

@book{Mattila1995,
  author = {Mattila, Pertti},
  title = {Geometry of Sets and Measures in Euclidean Spaces: Fractals and Rectifiability},
  series = {Cambridge Studies in Advanced Mathematics},
  volume = {44},
  publisher = {Cambridge University Press},
  year = {1995},
  doi = {10.1017/CBO9780511623813}
}

@article{Weedbrook2012,
  author = {Weedbrook, Christian and Pirandola, Stefano and Garc{\'i}a-Patr{\'o}n, Ra{\'u}l and Cerf, Nicolas J. and Ralph, Timothy C. and Shapiro, Jeffrey H. and Lloyd, Seth},
  title = {Gaussian quantum information},
  journal = {Reviews of Modern Physics},
  volume = {84},
  pages = {621--669},
  year = {2012},
  doi = {10.1103/RevModPhys.84.621},
  eprint = {1110.3234},
  archivePrefix = {arXiv}
}

@article{Branciard2012,
  title = {Classical Simulation of Entanglement Swapping with Bounded Communication},
  author = {Branciard, Cyril and Brunner, Nicolas and Buhrman, Harry and Cleve, Richard and Gisin, Nicolas and Portmann, Samuel and Rosset, Denis and Szegedy, Mario},
  journal = {Phys. Rev. Lett.},
  volume = {109},
  issue = {10},
  pages = {100401},
  numpages = {5},
  year = {2012},
  month = {Sep},
  publisher = {American Physical Society},
  doi = {10.1103/PhysRevLett.109.100401},
  url = {https://link.aps.org/doi/10.1103/PhysRevLett.109.100401}
}

@article{KullbackLeibler1951,
  author  = {Kullback, Solomon and Leibler, Richard A.},
  title   = {On Information and Sufficiency},
  journal = {The Annals of Mathematical Statistics},
  volume  = {22},
  number  = {1},
  pages   = {79--86},
  year    = {1951},
  doi     = {10.1214/aoms/1177729694}
}

@book{CoverThomas2006,
  author    = {Cover, Thomas M. and Thomas, Joy A.},
  title     = {Elements of Information Theory},
  edition   = {2},
  publisher = {John Wiley \& Sons},
  year      = {2006},
  isbn      = {978-0-471-24195-9},
  doi       = {10.1002/047174882X}
}

@article{Montina2011Approximate,
  author        = {Montina, Alberto},
  title         = {Approximate simulation of entanglement with a linear cost of communication},
  journal       = {Physical Review A},
  volume        = {84},
  number        = {4},
  pages         = {042307},
  year          = {2011},
  doi           = {10.1103/PhysRevA.84.042307},
  eprint        = {1107.4647},
  archivePrefix = {arXiv},
  primaryClass  = {quant-ph}
}

@article{Montina2013Exponential,
  author        = {Montina, Alberto},
  title         = {Exponential communication gap between weak and strong classical simulations of quantum communication},
  journal       = {Physical Review A},
  volume        = {87},
  number        = {4},
  pages         = {042331},
  year          = {2013},
  doi           = {10.1103/PhysRevA.87.042331},
  eprint        = {1301.3452},
  archivePrefix = {arXiv},
  primaryClass  = {quant-ph}
}

@article{ZartabEtAl2026,
  author        = {Zartab, Mani and Gasbarri, Giulio and Sent{\'i}s, Gael and Mu{\~n}oz-Tapia, Ramon},
  title         = {Prepare-and-measure and entanglement simulation beyond qubits},
  journal       = {Scientific Reports},
  volume        = {16},
  pages         = {21297},
  year          = {2026},
  doi           = {10.1038/s41598-026-51465-3},
  eprint        = {2508.02377},
  archivePrefix = {arXiv},
  primaryClass  = {quant-ph}
}

@article{Branciard2011,
  title = {Quantifying the Nonlocality of {Greenberger-Horne-Zeilinger} Quantum Correlations by a Bounded Communication Simulation Protocol},
  author = {Branciard, Cyril and Gisin, Nicolas},
  journal = {Phys. Rev. Lett.},
  volume = {107},
  issue = {2},
  pages = {020401},
  numpages = {4},
  year = {2011},
  month = {Jul},
  publisher = {American Physical Society},
  doi = {10.1103/PhysRevLett.107.020401},
  url = {https://link.aps.org/doi/10.1103/PhysRevLett.107.020401}
}

@ARTICLE{Brassard2016,
  author={Brassard, Gilles and Devroye, Luc and Gravel, Claude},
  journal={IEEE Transactions on Information Theory}, 
  title={Exact Classical Simulation of the Quantum-Mechanical {GHZ} Distribution}, 
  year={2016},
  volume={62},
  number={2},
  pages={876-890},
  doi={10.1109/TIT.2015.2504525}}

@article{Massar2001,
  title = {Classical simulation of quantum entanglement without local hidden variables},
  author = {Massar, Serge and Bacon, Dave and Cerf, Nicolas J. and Cleve, Richard},
  journal = {Phys. Rev. A},
  volume = {63},
  issue = {5},
  pages = {052305},
  numpages = {8},
  year = {2001},
  month = {Apr},
  publisher = {American Physical Society},
  doi = {10.1103/PhysRevA.63.052305},
  url = {https://link.aps.org/doi/10.1103/PhysRevA.63.052305}
}

@article{Buhrman2003,
  title = {Combinatorics and Quantum Nonlocality},
  author = {Buhrman, Harry and H{\o{}}yer, Peter and Massar, Serge and R{\"o}hrig, Hein},
  journal = {Phys. Rev. Lett.},
  volume = {91},
  issue = {4},
  pages = {047903},
  numpages = {4},
  year = {2003},
  month = {Jul},
  publisher = {American Physical Society},
  doi = {10.1103/PhysRevLett.91.047903},
  url = {https://link.aps.org/doi/10.1103/PhysRevLett.91.047903}
}

@article{Maudlin1992,
title = {Bell's Inequality, Information Transmission, and Prism Models},
 ISSN = {02708647},
 URL = {http://www.jstor.org/stable/192771},
  author = {Tim Maudlin},
 journal = {PSA: Proceedings of the Biennial Meeting of the Philosophy of Science Association},
 pages = {404--417},
 publisher = {[Cambridge University Press, Springer]},
 urldate = {2026-10-06},
 volume = {1992},
 year = {1992}
}

@article{Vertesi2009,
  title = {Lower bound on the communication cost of simulating bipartite quantum correlations},
  author = {V{\'e}rtesi, T. and Bene, E.},
  journal = {Phys. Rev. A},
  volume = {80},
  issue = {6},
  pages = {062316},
  numpages = {5},
  year = {2009},
  month = {Dec},
  publisher = {American Physical Society},
  doi = {10.1103/PhysRevA.80.062316},
  url = {https://link.aps.org/doi/10.1103/PhysRevA.80.062316}
}

@article{Degorre2007,
  title = {Classical simulation of traceless binary observables on any bipartite quantum state},
  author = {Degorre, Julien and Laplante, Sophie and Roland, J{\'e}r{\'e}mie},
  journal = {Phys. Rev. A},
  volume = {75},
  issue = {1},
  pages = {012309},
  numpages = {5},
  year = {2007},
  month = {Jan},
  publisher = {American Physical Society},
  doi = {10.1103/PhysRevA.75.012309},
  url = {https://link.aps.org/doi/10.1103/PhysRevA.75.012309}
}

@article{Pironio2003,
  title = {Violations of {Bell} inequalities as lower bounds on the communication cost of nonlocal correlations},
  author = {Pironio, Stefano},
  journal = {Phys. Rev. A},
  volume = {68},
  issue = {6},
  pages = {062102},
  numpages = {9},
  year = {2003},
  month = {Dec},
  publisher = {American Physical Society},
  doi = {10.1103/PhysRevA.68.062102},
  url = {https://link.aps.org/doi/10.1103/PhysRevA.68.062102}
}

@article{Cerf2000,
  title = {Classical Teleportation of a Quantum Bit},
  author = {Cerf, N. J. and Gisin, N. and Massar, S.},
  journal = {Phys. Rev. Lett.},
  volume = {84},
  issue = {11},
  pages = {2521--2524},
  numpages = {0},
  year = {2000},
  month = {Mar},
  publisher = {American Physical Society},
  doi = {10.1103/PhysRevLett.84.2521},
  url = {https://link.aps.org/doi/10.1103/PhysRevLett.84.2521}
}

@article{Regev2010,
author = {Regev, Oded and Toner, Ben},
title = {Simulating Quantum Correlations with Finite Communication},
journal = {SIAM Journal on Computing},
volume = {39},
number = {4},
pages = {1562-1580},
year = {2010},
doi = {10.1137/080723909},
URL = {https://doi.org/10.1137/080723909
    },
eprint ={https://doi.org/10.1137/080723909}
}

@article{Steiner2000,
title = {Towards quantifying non-local information transfer: finite-bit non-locality},
journal = {Physics Letters A},
volume = {270},
number = {5},
pages = {239-244},
year = {2000},
issn = {0375-9601},
doi = {https://doi.org/10.1016/S0375-9601(00)00315-7},
url = {https://www.sciencedirect.com/science/article/pii/S0375960100003157},
author = {Michael Steiner}
}

@article{Brunner2014,
  title = {Bell nonlocality},
  author = {Brunner, Nicolas and Cavalcanti, Daniel and Pironio, Stefano and Scarani, Valerio and Wehner, Stephanie},
  journal = {Rev. Mod. Phys.},
  volume = {86},
  issue = {2},
  pages = {419--478},
  numpages = {60},
  year = {2014},
  month = {Apr},
  publisher = {American Physical Society},
  doi = {10.1103/RevModPhys.86.419},
  url = {https://link.aps.org/doi/10.1103/RevModPhys.86.419}
}

@article{Janotta2014,
   title={Generalized probability theories: what determines the structure of quantum theory?},
   volume={47},
   ISSN={1751-8121},
   url={http://dx.doi.org/10.1088/1751-8113/47/32/323001},
   DOI={10.1088/1751-8113/47/32/323001},
   number={32},
   journal={Journal of Physics A: Mathematical and Theoretical},
   publisher={IOP Publishing},
   author={Janotta, Peter and Hinrichsen, Haye},
   year={2014},
   month=July, pages={323001} }

@article{Barrett2007,
  title = {Information processing in generalized probabilistic theories},
  author = {Barrett, Jonathan},
  journal = {Phys. Rev. A},
  volume = {75},
  issue = {3},
  pages = {032304},
  numpages = {21},
  year = {2007},
  month = {Mar},
  publisher = {American Physical Society},
  doi = {10.1103/PhysRevA.75.032304},
  url = {https://link.aps.org/doi/10.1103/PhysRevA.75.032304}
}

@article{Buhrman2010,
  title = {Nonlocality and communication complexity},
  author = {Buhrman, Harry and Cleve, Richard and Massar, Serge and de Wolf, Ronald},
  journal = {Rev. Mod. Phys.},
  volume = {82},
  issue = {1},
  pages = {665--698},
  numpages = {0},
  year = {2010},
  month = {Mar},
  publisher = {American Physical Society},
  doi = {10.1103/RevModPhys.82.665},
  url = {https://link.aps.org/doi/10.1103/RevModPhys.82.665}
}

@article{Brassard1999,
  title = {Cost of Exactly Simulating Quantum Entanglement with Classical Communication},
  author = {Brassard, Gilles and Cleve, Richard and Tapp, Alain},
  journal = {Phys. Rev. Lett.},
  volume = {83},
  issue = {9},
  pages = {1874--1877},
  numpages = {0},
  year = {1999},
  month = {Aug},
  publisher = {American Physical Society},
  doi = {10.1103/PhysRevLett.83.1874},
  url = {https://link.aps.org/doi/10.1103/PhysRevLett.83.1874}
}

@misc{YCL2026,
      title={Bounding Two-Way Average Communication Cost of Simulating Quantum Correlations}, 
      author={Kai-Siang Chen and Gelo Noel M. Tabia and Bo-An Tsai and Swati Kumari and Yeong-Cherng Liang},
      year={2026},
      eprint={2610.04648},
      archivePrefix={arXiv},
      primaryClass={quant-ph},
      url={https://arxiv.org/abs/2610.04648}, 
}

@Article{Brassard2018,
AUTHOR = {Brassard, Gilles and Devroye, Luc and Gravel, Claude},
TITLE = {Remote Sampling with Applications to General Entanglement Simulation},
JOURNAL = {Entropy},
VOLUME = {21},
YEAR = {2019},
NUMBER = {1},
ARTICLE-NUMBER = {92},
URL = {https://www.mdpi.com/1099-4300/21/1/92},
PubMedID = {33266808},
ISSN = {1099-4300},
DOI = {10.3390/e21010092}
}

@misc{deGoisSantosVieira2026,
  author        = {de Gois, Carlos and Santos, Thyago S. R. and Vieira, Carlos},
  title         = {Quantum communication and {Bell} nonlocality require infinite classical communication to simulate},
  year          = {2026},
  eprint        = {2609.04182},
  archivePrefix = {arXiv},
  primaryClass  = {quant-ph},
  doi           = {10.48550/arXiv.2609.04182},
  url           = {https://arxiv.org/abs/2609.04182},
  note          = {arXiv:2609.04182 [quant-ph], version 1}
}

@article{TonerBacon2003,
  author  = {Toner, Benjamin F. and Bacon, Dave},
  title   = {Communication Cost of Simulating {B}ell Correlations},
  journal = {Physical Review Letters},
  volume  = {91},
  number  = {18},
  pages   = {187904},
  year    = {2003},
  doi     = {10.1103/PhysRevLett.91.187904},
  eprint  = {quant-ph/0304076},
  archivePrefix = {arXiv}
}

@article{RennerTavakoliQuintino2023,
  author  = {Renner, Martin J. and Tavakoli, Armin and Quintino, Marco T{\'u}lio},
  title   = {Classical Cost of Transmitting a Qubit},
  journal = {Physical Review Letters},
  volume  = {130},
  number  = {12},
  pages   = {120801},
  year    = {2023},
  doi     = {10.1103/PhysRevLett.130.120801},
  eprint  = {2207.02244},
  archivePrefix = {arXiv},
  primaryClass  = {quant-ph}
}

\appendix
\begin{widetext}

\section{Finite average classical dimension is impossible for \texorpdfstring{$d\geq4$}{d >= 4}}
\label{sec:finite-average-dimension-impossible}

We first recall the fixed-fidelity estimate proved as Lemma~B.2 in
Ref.~\cite{deGoisSantosVieira2026}.  For $0\leq t<1$, let $\eta_t$
be the unitarily invariant probability
measure on pairs $(x,y)\in\mathbb{CP}^{3}\times\mathbb{CP}^{3}$ with
fixed fidelity
\begin{equation}
    \abs{\braket{x}{y}}^2=t.
\end{equation}
Equivalently, if $W\in U(4)$ is Haar distributed and
$\ket{0},\ket{1}$ are fixed orthogonal vectors, then $\eta_t$ is the
law of
\begin{equation}
    x=[W\ket{0}],
    \qquad
    y=\left[W\left(\sqrt t\ket{0}+\sqrt{1-t}\ket{1}\right)\right].
    \label{eq:fixed-fidelity-pair}
\end{equation}
For measurable functions $A,B:\mathbb{CP}^{3}\to[0,1]$, define
\begin{equation}
    a_{A,B}(t)
    :=
    \int_{\mathbb{CP}^{3}\times\mathbb{CP}^{3}}
    A(x)B(y)\,\mathrm d\eta_t(x,y).
    \label{eq:fixed-fidelity-average}
\end{equation}

\begin{lemma}[Fixed-fidelity zero-slope lemma]
\label{lem:fixed-fidelity-zero-slope}
Let $A,B:\mathbb{CP}^{3}\to[0,1]$ be measurable.  If
$a_{A,B}(0)=0$, then
\begin{equation}
    \lim_{t\to0^+}\frac{a_{A,B}(t)}{t}=0,
    \label{eq:zero-slope-limit}
\end{equation}
and, for $0<t\leq1/2$,
\begin{equation}
    0\leq\frac{a_{A,B}(t)}{t}\leq8.
    \label{eq:zero-slope-domination}
\end{equation}
\end{lemma}

The lemma is quoted from Ref.~\cite{deGoisSantosVieira2026}; its proof
is not repeated here. As we now show, this lemma rules out not only a fixed
finite message alphabet, as shown in the latter reference, but also a random
alphabet with finite mean dimension.

\begin{theorem}[No finite average classical dimension]
For $d\geq4$, every measurable one-way exact simulation of pure
states and binary projective measurements with a rank-one effect
has infinite average classical dimension:
\begin{equation}
 \Exp[\AlpS_\lambda]=\infty.
    \label{eq:infinite-average-classical-dimension}
\end{equation}
\end{theorem}

\begin{proof}
It suffices to prove the claim for $d=4$.  For $d>4$, restrict the
preparations to a fixed four-dimensional subspace and take binary
measurements whose rank-one effect is supported on that subspace.

Group the shared seeds by their finite alphabet size. Equivalently,
let $K$ be countable, with $\Prob(K=k)=w_k$, and let the alphabet
dimension conditional on $K=k$ be $D_k<\infty$. Then
$\Exp[\AlpS_\lambda]=\sum_{k\geq1}w_kD_k$.
Assume that
\begin{equation}
    \sum_{k\geq1}w_kD_k<\infty.
    \label{eq:finite-mean-dimension-assumption}
\end{equation}
Conditionally on $K=k$, let $\lambda\in\Lambda_k$ denote all remaining
shared randomness, with probability measure $\mu_k$.  For a pure preparation
$x\in\mathbb{CP}^{3}$ and the binary measurement
$\{\ketbra{y},\one-\ketbra{y}\}$, write
\begin{equation}
    A_{\lambda,k,m}(x)
    :=\mathbf{1}_{\{M_{\lambda,k}(x)=m\}},
    \qquad
    B_{\lambda,k,m}(y)
    :=q(0\mid y,m,\lambda,k).
\end{equation}
In other words, $A_{\lambda,k,m}(x)\in \{0,1\}$ equals one if and only if $m$ is the message that Alice sends to Bob when the shared randomness $(\lambda,k)$ is sampled and the input $x$ is given. On the other hand $B_{\lambda,k,m}(y)$ represents the probability of Bob's decoder to output $0$ given Alice's message, the POVM measurement choice (through $\ket{y}$) and the shared randomness.  
We assume as in \cite{deGoisSantosVieira2026} that these are both measurable $[0,1]$-valued functions. If the simulation is exact then it has to satisfy
\begin{equation}
    \abs{\braket{x}{y}}^2
    =
    \sum_{k\geq1}w_k
    \int_{\Lambda_k}
    \sum_{m}=1^{D_k}
    A_{\lambda,k,m}(x)B_{\lambda,k,m}(y)
    \,\mu_k(\dmu\lambda).
    \label{eq:variable-dimension-overlap-decomposition}
\end{equation}
For each $(\lambda,k,m)$, let
\begin{equation}
    f_{\lambda,k,m}(t)
    :=a_{A_{\lambda,k,m},B_{\lambda,k,m}}(t).
\end{equation}
Averaging \cref{eq:variable-dimension-overlap-decomposition} with
respect to $\eta_t$ and using Tonelli's theorem yields
\begin{equation}
    t
    =
    \sum_{k\geq1}w_k
    \int_{\Lambda_k}
    \sum_{m=1}^{D_k}f_{\lambda,k,m}(t)
    \,\mu_k(\dmu\lambda)
    =
    \sum_{k\geq1}w_k
    \sum_{m=1}^{D_k}f_{k,m}(t),
    \label{eq:variable-dimension-fixed-fidelity-identity}
\end{equation}
where
\begin{equation}
    f_{k,m}(t):=
    \int_{\Lambda_k}f_{\lambda,k,m}(t)\,\mu_k(\dmu\lambda).
\end{equation}

Setting $t=0$ in
\cref{eq:variable-dimension-fixed-fidelity-identity}, non-negativity
implies, for every $k$ with $w_k>0$,
\begin{equation}
    f_{\lambda,k,m}(0)=0
    \qquad
    \text{for $\mu_k$-almost every $\lambda$ and every $m\in[D_k]$}.
\end{equation}
Indices with $w_k=0$ may be omitted. Thus we can apply
\cref{lem:fixed-fidelity-zero-slope} and obtain
\begin{equation}\label{Eq:Lemma-application}
    \lim_{t\to0^+}\frac{f_{\lambda,k,m}(t)}{t}=0,
    \qquad
    0\leq\frac{f_{\lambda,k,m}(t)}{t}\leq8
    \quad (0<t\leq\tfrac12).
\end{equation}
By dominated convergence with respect to $\mu_k$ and finiteness
of $D_k$, this implies the analogous statement for
$f_k(t):=\sum_{m=1}^{D_k}f_{k,m}(t)$:
\begin{equation} \label{Eq:Lemma-application2}
    \lim_{t\to0^+}\frac{f_k(t)}{t}=0,
    \qquad
    0\leq\frac{f_k(t)}{t}\leq8D_k
    \quad (0<t\leq\tfrac12).
\end{equation}

We then use the crucial assumption $\sum_k w_kD_k<\infty$
to obtain the summable bound
\begin{equation}
    0\leq w_k\frac{f_k(t)}{t}\leq8w_kD_k,
    \qquad
    \sum_{k\geq1}8w_kD_k<\infty.
\end{equation}
Dominated convergence therefore allows the limit to pass through
the sum over $k$, giving
\begin{equation}
   \lim_{t\to0^+} \sum_{k\geq1}w_k\frac{f_k(t)}{t}=
    \sum_{k\geq1}w_k\,0=0
\end{equation}
contradicting \cref{eq:variable-dimension-fixed-fidelity-identity}.
Therefore \cref{eq:infinite-average-classical-dimension} is necessary.
\end{proof}
 
As the point is rather subtle, we write down explicitly the mechanism through which our protocol for the exact simulation of quantum probability distribution in arbitrary dimension escapes this obstruction. All the passages leading to Eq.\eqref{Eq:Lemma-application2} hold without assuming that $\Exp[\AlpS_\lambda]<\infty$, thus the functions $f_{k}(t)$ still satisfy the latter. However, in our case $D_k\asymp2^{n_dk}$ and $w_k=2^{-k}$,
with $n_d=2d-2\geq2$, so that
$\sum_k w_kD_k=\infty$.
To illustrate why the bound $f_k(t)/t\leq8D_k$ is insufficient
in this regime, take $D_k=2^{n_dk}$ and consider the following
functions on $[0,\tfrac12]$:
\begin{equation}
    f_k(t):=
    \begin{cases}
        2^{k+1}t-1,
        & 2^{-(k+1)}\leq t\leq2^{-k},\\
        2-2^kt,
        & 2^{-k}<t\leq2^{-(k-1)},\\
        0,
        & \text{otherwise}.
    \end{cases}
\end{equation}
Each function is continuous, takes values in $[0,1]$, and vanishes
in a neighborhood of $t=0$. In particular,
\begin{equation}
    \lim_{t\to0^+}\frac{f_k(t)}{t}=0,
    \qquad
    0\leq\frac{f_k(t)}{t}\leq2^k\leq8D_k.
\end{equation}
Thus the functions satisfy the same zero-slope condition and
upper bound as in the proof. 
Nevertheless, for $t\in[2^{-(n+1)},2^{-n}]$, only the terms
with indices $n$ and $n+1$ can be non zero, and thus
\begin{equation}
    \begin{aligned}
        \sum_{k\geq1}w_kf_k(t)
        &=
        2^{-n}(2^{n+1}t-1)
        +2^{-(n+1)}(2-2^{n+1}t)=t.
    \end{aligned}
\end{equation}
This illustrates how having a diverging average classical dimension removes the obstruction.

\section{Mass-deficit Minkowski dimension and nested partitions}
\label{app:dimension-partitionability}

\begin{lemma}
\label{lem:dimension-partitionability}
The pair $(X,\cF)$ has finite mass-deficit Minkowski dimension if and
only if $X$ is well-partitionable with respect to $\cF$. Further, for any $\nu>4d_{\mathrm{mdM}}(X,\cF)$, one can
construct nested partitions $\{\cC_r\}_{r\geq0}$ and constants
$a,C>0$ such that, for every integer $r\geq0$,
\begin{equation}
    |\cC_r|\leq2^{a+\nu r},
    \qquad
    \kappa_{\cF}(\cC_r)\leq C2^{-r}.
\end{equation}
\end{lemma}

\begin{proof}
Suppose first that nested partitions satisfy
\cref{eq:partitionability-bounds}. For sufficiently small
$\epsilon>0$, choose $r=\lceil\log(C/\epsilon)\rceil$. The cells
of $\cC_r$ cover $X$ with mass deficit at most $\epsilon$, so
\begin{equation}
    N_{\cF}(\epsilon)\leq|\cC_r|
    \leq2^a\left(\frac{2C}{\epsilon}\right)^\nu.
\end{equation}
Taking the upper limit in \cref{eq:mass-deficit-minkowski-dimension}
gives $d_{\mathrm{mdM}}(X,\cF)\leq\nu$.

Conversely, suppose $d_{\mathrm{mdM}}(X,\cF)<\infty$ and fix
$\alpha>d_{\mathrm{mdM}}(X,\cF)$. There is an integer $m\geq1$
such that
\begin{equation}
    N_{\cF}(2^{-n})\leq2^{\alpha n}\qquad(n\geq m).
    \label{eq:covering-polynomial-bound}
\end{equation}
A finite cover can be made into a partition by successively
subtracting the preceding sets and discarding empty pieces.
Monotonicity of the mass deficit under inclusion preserves the bound.

Set $m_k=m2^k$, for $k\geq0$, and choose finite partitions
$\mathcal Q_k$ satisfying
\begin{equation}
    |\mathcal Q_k|\leq2^{\alpha m_k},
    \qquad\kappa_{\cF}(\mathcal Q_k)\leq2^{-m_k}.
\end{equation}
Let $\mathcal R_k$ be their common refinement through level $k$:
\begin{equation}
    \mathcal R_k:=
    \{Q_0\cap\cdots\cap Q_k:Q_i\in\mathcal Q_i\}
    \setminus\{\emptyset\}.
\end{equation}
These partitions are nested and satisfy
\begin{equation}
    |\mathcal R_k|
    \leq\prod_{i=0}^k|\mathcal Q_i|
    \leq2^{\alpha(2m_k-m)},
    \qquad\kappa_{\cF}(\mathcal R_k)\leq2^{-m_k}.
    \label{eq:sparse-common-refinements}
\end{equation}
The exponentially spaced indices $m_k$ keep the cost of taking
common refinements exponential in the final index.

For $r\geq1$, let $k(r)$ be the smallest $k$ with $m_k\geq r$,
and set $\cC_r:=\mathcal R_{k(r)}$. Then $k(r)$ is nondecreasing,
so these partitions are nested. Since $m_{k(r)}\geq r$ and
$m_{k(r)}\leq m+2r$, \cref{eq:sparse-common-refinements} gives
\begin{equation}
    |\cC_r|\leq2^{\alpha m+4\alpha r},
    \qquad\kappa_{\cF}(\cC_r)\leq2^{-r}.
    \label{eq:dimension-to-partition-bounds}
\end{equation}

Set $\cC_0=\{X\}$. Since $\kappa_{\cF}(X)\leq1$,
\cref{eq:partitionability-bounds} holds with
$a=\alpha m$, $\nu=4\alpha$, and $C=1$, completing the proof.
\end{proof}

\section{Nested partitions of complex projective space}
\label{app:quantum-partitions}
We construct the partitions used in the quantum simulation. We emphasize that this particular choice of nested partition is not optimal for our scope. It has been chosen instead for its simplicity.   
Throughout this appendix, $d\geq2$ and $\diam_D$ denotes diameter
in trace distance, as defined in \cref{sec:nested-colorings}.

\begin{lemma}[Nested coordinate partitions]
\label{lem:nested-partitions}
There are finite nested partitions
$\{\cC_r\}_{r\geq0}$ of $\Pure$ such that
\begin{equation}
    |\cC_r|\leq d\,2^{n_dr},\qquad
    \diam_D(\cC_r)\leq C_d2^{-r},
    \label{eq:dyadic-coordinate-partition-size}
\end{equation}
where $n_d=2d-2$ and $C_d=4\sqrt{n_d}$.
Every cell at a positive level has at most $2^{n_d}$ children.
\end{lemma}

The idea applies to any compact smooth manifold of real dimension
$n$. Compactness provides finitely many coordinate neighborhoods
whose closures lie inside larger charts. Their coordinate images
fit inside bounded cubes. Repeatedly halving every coordinate
interval creates at most $2^n$ children per cell and halves the
coordinate diameter. The inverse chart maps have bounded
derivatives on these compact regions, so cell diameters in a fixed
Riemannian metric are bounded by a constant times $2^{-r}$.

\begin{proof}
For a ray $\ketbra{\psi}$, let $\jmath(\ketbra{\psi})$ be the smallest index
at which $|\psi_j|$ is maximal. The sets
\begin{equation}
    X_j:=\{\ketbra{\psi}:\jmath(\ketbra{\psi})=j\},
    \qquad j\in[d],
    \label{eq:projective-chart-partition}
\end{equation}
form a partition of $\Pure$.
On $X_j$, we have $\psi_j\neq0$, and the ratios
$\psi_i/\psi_j$ determine the ray independently of its
representative. Separating their real and imaginary parts gives
the $n_d$ real coordinates
\begin{equation}
    \begin{aligned}
        u_i^{(j)}&:=\operatorname{Re}\frac{\psi_i}{\psi_j},
        &v_i^{(j)}&:=\operatorname{Im}\frac{\psi_i}{\psi_j},
        \qquad i\neq j,\\
        z^{(j)}(\ketbra{\psi})
        &:=\bigl(u_i^{(j)},v_i^{(j)}\bigr)_{i\neq j}
        \in[-1,1]^{n_d}.
    \end{aligned}
    \label{eq:inhomogeneous-coordinates}
\end{equation}
Indeed, $|\psi_i/\psi_j|\leq1$, so every real coordinate
lies in $[-1,1]$.

Let $\mathcal D_r$ be the partition of $[-1,1]^{n_d}$ into
$2^{n_dr}$ cubes of side $2^{1-r}$, obtained by bisecting
every coordinate interval $r$ times.
To assign boundary points consistently, use half-open intervals,
with the final interval including the outer endpoint.
Set $\cC_0=\{\Pure\}$ and, for $r\geq1$, define
\begin{equation}
    \cC_r:=
    \left\{
        X_j\cap(z^{(j)})^{-1}(Q):
        j\in[d],\ Q\in\mathcal D_r
    \right\}\setminus\{\emptyset\}.
    \label{eq:coordinate-cell-definition}
\end{equation}

Each partition refines the preceding one. At positive levels, each
cell has at most $2^{n_d}$ children. There are $d$ chart labels and
$2^{n_dr}$ cubes per chart, proving $|\cC_r|\leq d2^{n_dr}$.
The cardinality bound also holds at $r=0$.

It remains to estimate the cell diameter.
For a real coordinate vector $z=(u_i,v_i)_{i\neq j}$, the
associated normalized representative is
\begin{equation}
    \varphi_j(z):=
    \frac{e_j+\sum_{i\neq j}(u_i+\mathrm{i}v_i)e_i}
         {\sqrt{1+\|z\|_2^2}}.
    \label{eq:normalized-chart-map}
\end{equation}
Normalization of vectors of norm at least one is
$2$-Lipschitz, and pure-state trace distance is at most the
Euclidean distance between unit-vector representatives. Hence
\begin{equation}
    D(\varphi_j(z),\varphi_j(z'))
    \leq2\|z-z'\|_2.
    \label{eq:chart-lipschitz-bound}
\end{equation}
Two points in the same cube satisfy
$\|z-z'\|_2\leq\sqrt{n_d}\,2^{1-r}$, giving the diameter upper bound in 
\cref{eq:dyadic-coordinate-partition-size} with
$C_d=4\sqrt{n_d}$. At the root, $\diam_D(\Pure)=1\leq C_d$.
\end{proof}

\section{Continuous outcome spaces}
\label{app:continuous-outcomes}

The nested-partition construction extends to continuous outcomes by
replacing probability vectors with probability densities and sums over
outcomes with integrals. The lower envelopes are understood
up to sets of measure zero. With these adaptations, the
same overlap estimate and recursive decoder prove exact simulation.

\subsection{Target distributions and exactness}

Let $(\mathcal W,\Sigma,\tau)$ be a $\sigma$-finite measure space. Define
\begin{equation}
    \Delta_\tau:=\left\{g\in L^1(\tau):g\geq0,\quad
    \int_{\mathcal W}g(w)\,\mathrm d\tau(w)=1\right\}.
    \label{eq:continuous-simplex}
\end{equation}
The target family $\cF$ now consists of maps $f:X\to\Delta_\tau$.
Writing $f_w(x)$ for a representative of the density $f(x)$, the
probability of an outcome event $A\in\Sigma$ is
\begin{equation}
    P_f(A\mid x):=\int_A f_w(x)\,\mathrm d\tau(w).
    \label{eq:continuous-target-law}
\end{equation}
By the Radon--Nikodym theorem\footnote{D.~H.~Fremlin,
\emph{Measure Theory}, Vol.~2, 232F;
\url{https://www1.essex.ac.uk/maths/people/fremlin/chap23.pdf}.}, this formulation is equivalent to
requiring that all the target laws $P_f(\cdot\mid x)$ be absolutely
continuous with respect to $\tau$. Densities that agree
$\tau$-almost everywhere represent the same density.

Alice still receives $x\in X$, Bob receives $f\in\cF$, and the shared random
variable $\lambda\sim\mu$ selects a finite message alphabet. For our
construction, Bob's conditional output laws have densities
$q(\cdot\mid f,\lambda,m)\in\Delta_\tau$. Exactness means that,
for every $x\in X$ and $f\in\cF$,
\begin{equation}
    f_w(x)=\int_\Omega
    q\bigl(w\mid f,\lambda,M_\lambda(x)\bigr)\,
    \mathrm d\mu(\lambda)
    \qquad\text{for $\tau$-almost every }w.
    \label{eq:continuous-exactness}
\end{equation}
For each fixed $x,f$, we use representatives for which the integrand
is jointly measurable in $(\lambda,w)$. This is always satisfied for the
countable shared randomness used in our construction.
The communication cost is still $\Exp[\log \AlpS_\lambda]$ from {\cref{eq:expected-bit-cost}}.

\subsection{Mass deficit, covering dimension, and the decoder}

For nonnegative functions write $g\preceq h$ if $g\leq h$
$\tau$-almost everywhere, and write $|g|:=\int g\,\mathrm d\tau$.
For a nonempty region $R\subseteq X$, define
\begin{equation}
    \ell^f(R):=\operatorname*{ess\,inf}_{x\in R}f(x).
    \label{eq:continuous-lower-envelope}
\end{equation}
This uses the standard essential infimum of a
\emph{family of measurable functions}, understood in the
almost-everywhere order.\footnote{See D.~H.~Fremlin,
\emph{Measure Theory}, Vol.~2, \S241, particularly
241E--241G for the order structure and existence of suprema,
and 241Ye for the countable-subfamily property.
Available at
\url{https://www1.essex.ac.uk/maths/people/fremlin/chap24.pdf}.}
More explicitly, the essential infimum of a family
$\{h_i\}_{i\in I}$ of nonnegative measurable functions is a
measurable function $H$, possibly taking the value $+\infty$,
defined by the requirements:
\begin{align}
    &h_i\geq H
      \quad\text{$\tau$-almost everywhere, for every }i\in I,
      \label{eq:essential-supremum-upper-bound}\\
    &\left[
      h_i \geq G
      \quad\text{$\tau$-almost everywhere, for every }i\in I
      \right]
      \ \Longrightarrow\
      G\leq H
      \quad\text{$\tau$-almost everywhere}
      \label{eq:essential-supremum-minimality}
\end{align}
for every measurable lower bound $G$. The essential infimum is unique up to
$\tau$-null sets.

For example, on $[0,1]$ with Lebesgue measure, the functions $g_t(w):=1-\mathbf{1}_{\{t\}}(w)$
have essential infimum one, although their pointwise infimum
is identically zero.
Thus essential envelopes ignore changes of each function on
null sets, as required when working with densities.

Under $\sigma$-finiteness, these envelopes exist and can be
obtained from countable subfamilies. In particular,
$\ell^f(R)$ is the greatest almost-everywhere lower bound of
all $f(x)$. It is independent of density representatives and is
measurable even for an uncountable set of inputs.

The mass deficit is
\begin{equation}
    \kappa_{\cF}(R):=\sup_{f\in\cF}
        \left(1-\int_{\mathcal W}\ell_w^f(R)\,\mathrm d\tau(w)\right).
    \label{eq:continuous-family-mass-deficit}
\end{equation}
Taking $\tau$ to be counting measure on $[s]$ gives
\cref{eq:family-mass-deficit}. For every $x\in R$,
$0\preceq\ell^f(R)\preceq f(x)$, so the lower envelope is a
nonnegative subdensity of mass at most one. Consequently,
$0\leq\kappa_{\cF}(R)\leq1$ and
\begin{equation}
    |\ell^f(R)|\geq1-\kappa_{\cF}(R).
    \label{eq:continuous-overlap}
\end{equation}
If $R\subseteq P$, then
$\ell^f(P)\preceq\ell^f(R)$ and
$\kappa_{\cF}(R)\leq\kappa_{\cF}(P)$.

Define $N_{\cF}(\epsilon)$ and $d_{\mathrm{mdM}}(X,\cF)$ as in
\cref{def:mass-deficit-minkowski}, using
\cref{eq:continuous-family-mass-deficit}. Nested partitions and
well-partitionability are defined by
\cref{def:nested-partitions,def:well-partitionable}.
The equivalence in \cref{lem:dimension-partitionability} holds:
its covering and refinement argument uses monotonicity under
inclusion and the bound $\kappa_{\cF}(X)\leq1$.

\begin{theorem}[Exact simulation with continuous outcomes]
\label{thm:continuous-simulation}
If $(X,\cF)$ has finite mass-deficit Minkowski dimension with respect to
\cref{eq:continuous-family-mass-deficit}, then $\cF$ admits a one-way
exact simulation with finite expected communication and an
exponential tail for the message length.
\end{theorem}

\begin{proof}
The proof of \cref{thm:general-simulation} applies in
$L^1(\tau)$, replacing sums over outcomes with integrals and
coordinatewise inequalities with almost-everywhere inequalities.
The overlap estimate \cref{eq:continuous-overlap} gives the
cumulative invariants of the recursive decoder. The telescoping
argument gives convergence to $f(x)$ in $L^1(\tau)$ and hence
\cref{eq:continuous-exactness}. The shared resolution law and
the bounds on the message alphabets give
\cref{eq:general-average-bits,eq:general-bit-tail}.
\end{proof}

\section{Approximate classical simulation} \label{app:ApproxSimulation}
Here we give a proof of Proposition \ref{prop:approx-lower-envelope}, for the reader's convenience we first re-state the proposition. 
\begin{proposition}
For every $d\geq2$ and all sufficiently large $r$, the protocol using $\cC_{r}$ partitions of appendix \ref{app:quantum-partitions} for the encoding, and the normalized lower envelope for the decoding 
satisfies
\begin{equation}
    D\!\left(
    q(\mathsf E,\psi)
    \,\middle\|\,
    p(\mathsf E,\psi)
    \right)
    \leq\frac{K_d}{N^{\alpha_d}},
    \qquad
    \alpha_d=\frac1{d-1},
    \label{eq:approx-relative-entropy-bound}
\end{equation}
uniformly over all input states and POVMs, where $K_d$ depends
only on $d$.
\end{proposition}
\begin{proof}
Recall that the partitions in
Appendix~\ref{app:quantum-partitions} satisfy
\begin{equation}
    N\leq d\,2^{n_dr},
    \qquad
    \diam_D(\cC_r)\leq
    \delta_r:=4\sqrt{n_d}\,2^{-r}.
\end{equation}
Choose $r$ such that $d\delta_r\leq\tfrac12$.
By \cref{lem:Born-stability},
\begin{equation}
    \abs{\ell(\mathsf E,R_m)}
    \geq1-d\delta_r\geq\tfrac12.
\end{equation}
In particular, the decoder is well-defined.

Fix $\psi$ and $\mathsf E=(E_b)_b$, and let $R=R_{M(\psi)}$
be the cell containing $\psi$.
We omit the dependence on the state, measurement, and cell when
clear from the context, writing
\begin{equation}
    p:=p(\mathsf E,\psi),\qquad
    q:=q(\mathsf E,\psi),\qquad
    \ell:=\ell(\mathsf E,R).
\end{equation}
Set
\begin{equation}
    e_b:=p_b-\ell_b,
    \qquad
    t:=\sum_b e_b=1-\abs{\ell}.
\end{equation}
Then $0\leq t\leq d\delta_r$, and Bob's output distribution satisfies
$q_b=(p_b-e_b)/(1-t)$.

We first bound the deficits $e_b$ using the Born-rule structure.
For $\phi\in R$, choose unit-vector representatives whose inner
product is nonnegative. Then
\begin{equation}
    \|\psi-\phi\|_2^2
    =2\bigl(1-|\braket{\psi}{\phi}|\bigr)
    \leq2D(\psi,\phi)^2
    \leq2\delta_r^2.
\end{equation}
Consequently,
\begin{align}
    \left|\sqrt{p_b}
          -\sqrt{p_b(\mathsf E,\phi)}\right|
    &=
    \left|\|\sqrt{E_b}\psi\|_2
          -\|\sqrt{E_b}\phi\|_2\right|\\
    &\leq\sqrt{2\|E_b\|_\infty}\,\delta_r.
\end{align}
Taking the infimum over $\phi\in R$ yields
\begin{equation}
    e_b
    =(\sqrt{p_b}-\sqrt{\ell_b})
     (\sqrt{p_b}+\sqrt{\ell_b})
    \leq2\sqrt{2\|E_b\|_\infty p_b}\,\delta_r.
\end{equation}
Since $\sum_b\|E_b\|_\infty\leq\sum_b\operatorname{tr}E_b=d$,
we obtain
\begin{equation}
    \sum_{b:p_b>0}\frac{e_b^2}{p_b}
    \leq8d\,\delta_r^2.
    \label{eq:approx-weighted-deficit}
\end{equation}
If $p_b=0$, then $\ell_b=e_b=q_b=0$, so such outcomes
do not contribute to the divergence.

Using $\ln x\leq x-1$, followed by
\cref{eq:approx-weighted-deficit}, gives
\begin{align}
    D(q\|p)
    &\leq\frac1{\ln2}
      \sum_{b:p_b>0}\frac{(q_b-p_b)^2}{p_b}\\
    &=\frac1{(1-t)^2\ln2}
      \left(\sum_{b:p_b>0}\frac{e_b^2}{p_b}-t^2\right)\\
    &\leq\frac{8d\,\delta_r^2}
                 {(1-d\delta_r)^2\ln2}
    \leq\frac{32d\,\delta_r^2}{\ln2}.
\end{align}
Finally, $N\leq d\,2^{n_dr}$ implies
\begin{equation}
    \delta_r^2
    =16n_d\,2^{-2r}
    \leq16n_d\left(\frac dN\right)^{2/n_d}.
\end{equation}
Thus \cref{eq:approx-relative-entropy-bound} holds with
\begin{equation}
    \alpha_d=\frac2{n_d}=\frac1{d-1},
    \qquad
    K_d:=\frac{512n_d\,d^{1+2/n_d}}{\ln2}.
\end{equation}
\end{proof}

%%%%%%%%%%%%%%%%%%%%%%%%%%%%%%%%%%%%%%%%%%%%%%%%%%%%%%%%%%%%%
%%%%%%%%%%%%%%%%%%%%%%%%%%%%%%%%%%%%%%%%%%%%%%%%%%%%%%%%%%%%%

\section{Multimode Gaussian states and general-dyne measurements}
\label{app:multimode-gaussian}

In this appendix we collect the description of $n$-mode Gaussian states and Gaussian (general-dyne) measurements that are needed for the simulation model of Sec.~\ref{sec:gaussian}, and we prove the bound \cref{eq:gaussian-average-bits}. We follow the conventions used in Ref.~\cite{ValcarceEtAl2022} (Appendix A), where a reader can find a more detailed exposition. For general reviews on Gaussian states, measurements and quantum information we refer the reader to Refs.~\cite{Weedbrook2012,AdessoRagyLee2014}. 

\subsection{Quadratures, displacement vector and covariance matrix}

Consider $n$ bosonic modes with ladder operators $a_i,a_i^\dagger$,
$[a_i,a_j^\dagger]=\delta_{ij}$, and the dimensionless quadrature operators
\begin{equation}
    \hat x_i=\frac{a_i^\dagger+a_i}{2},\qquad
    \hat p_i=\ii\,\frac{a_i^\dagger-a_i}{2},
   \end{equation}
collected in the vector
$\mathbf q=(\hat x_1,\hat p_1,\dots,\hat x_n,\hat p_n)^T=(q_1,\dots,q_{2n})^T$.
The canonical commutation relations read
\begin{equation}
    [q_k,q_l]=\frac{\ii}{2}\,\Omega_{kl},
    \qquad
    \Omega=\one_n\otimes
    \begin{pmatrix} 0&1\\-1&0\end{pmatrix}.
    \label{eq:mm-ccr}
\end{equation}
The first two moments of the quadratures on a state $\rho$ define the
displacement (column) vector $\bm\mu\in\mathbb R^{2n}$ and the covariance
matrix $\Sigma=\Sigma^T\in\mathbb R^{2n\times 2n}$,
\begin{equation}
    \mu_k=\tr(\rho\,q_k),\qquad
    \Sigma_{kl}=\tfrac12\tr\!\big(\rho\,\{q_k,q_l\}\big)-\mu_k\mu_l .
    \label{eq:mm-moments}
\end{equation}
With these conventions the vacuum $\ket0^{\otimes n}$ has
$(\bm\mu,\Sigma)=(\vec 0,\tfrac14\one_{2n})$, and a coherent state
$\ket{\bm\alpha}=\bigotimes_{i=1}^n\ket{\alpha_i}$ has
$(\bm\mu, \Sigma)=(\tilde{\bm\alpha},\tfrac14\one_{2n})$, where $ \tilde{\bm\alpha}:=
    \big(\mathrm{Re}\,\alpha_1,\mathrm{Im}\,\alpha_1,\dots,
    \mathrm{Re}\,\alpha_n,\mathrm{Im}\,\alpha_n\big)^T.$

Any real symmetric $\Sigma$ is the covariance matrix of a quantum state
if and only if it satisfies the uncertainty relation~\cite{SimonMukundaDutta1994}, in our normalization
\begin{equation}
    \Sigma+\frac{\ii}{4}\,\Omega\geq0 .
    \label{eq:mm-uncertainty}
\end{equation}
In particular $\Sigma$ is positive semidefinite. Indeed, for a real vector $\bm v$ one has $\bm v^T\Omega\bm v=0$, since $\Omega$ is antisymmetric, so
\cref{eq:mm-uncertainty} gives $\bm v^T\Sigma\bm v\geq0$. In fact $\Sigma$ is
positive definite, see \cref{eq:mm-ab} below.

\paragraph{Wigner function.}
A state $\rho$ is Gaussian if its Wigner function
$W_\rho(\bm\alpha):=\tr\big[\rho\bigotimes_{i=1}^n\delta(a_i-\alpha_i)\big]$
is the density of a multivariate normal distribution,
\begin{equation}
    W_\rho(\bm\alpha)=\mathrm N(\tilde{\bm\alpha};\bm\mu,\Sigma),
    \qquad
    \mathrm N(\bm y;\bm\mu,\Sigma):=
    \frac{\exp\!\big[-\tfrac12(\bm y-\bm\mu)^T\Sigma^{-1}(\bm y-\bm\mu)\big]}
         {\sqrt{\det(2\pi\Sigma)}},
    \label{eq:mm-wigner}
\end{equation}
normalized as $\int\mathrm d^{2n}\bm\alpha\,W_\rho(\bm\alpha)=1$ with
$\mathrm d^{2n}\bm\alpha=\prod_i\mathrm d\,\mathrm{Re}\,\alpha_i\,
\mathrm d\,\mathrm{Im}\,\alpha_i$~\cite{ValcarceEtAl2022,VogelWelsch2006}. A Gaussian state is thus faithfully
described by the pair $(\bm\mu,\Sigma)$, i.e.\ by $2n^2+3n$ real
parameters, and we write $\rho\simeq(\bm\mu,\Sigma)$. Conversely, every pair
$(\bm\mu,\Sigma)$ with $\Sigma$ satisfying \cref{eq:mm-uncertainty} defines a
Gaussian state. In this normalization the overlap of two states reads
\begin{equation}
    \tr(\rho\,\sigma)=\pi^n\int\mathrm d^{2n}\bm\alpha\;
    W_\rho(\bm\alpha)\,W_\sigma(\bm\alpha),
    \label{eq:mm-overlap}
\end{equation}which follows from the inversion formula $\rho=\pi^n\int\mathrm d^{2n}\bm\alpha\,W_\rho(\bm\alpha)\bigotimes_i\delta(a_i-\alpha_i)$~\cite{VogelWelsch2006,ValcarceEtAl2022}.

\paragraph{Energy.}
The total photon number operator $\hat N=\sum_ia_i^\dagger a_i$ satisfies
$\hat N=\sum_i(\hat x_i^2+\hat p_i^2)-\tfrac n2$. Hence for any state
\begin{equation}
    \tr(\rho\hat N)=\|\bm\mu\|^2+\tr\Sigma-\frac n2 .
    \label{eq:mm-energy}
\end{equation}
For instance, a single-mode squeezed vacuum with squeezing parameter $r$
has $\tr\Sigma=\tfrac12\cosh 2r$ and mean photon number $\sinh^2r$.

\subsection{Gaussian states with bounded energy}
\label{app:mm-energy-bounded}

The set of $n$-mode Gaussian states with mean photon number at most $E$ is
\begin{equation}
    X_E=\Big\{\rho\simeq(\bm\mu,\Sigma):\ \|\bm\mu\|^2+\tr\Sigma-\frac n2\leq E,
    \quad\Sigma+\frac\ii4\Omega\geq0\Big\}.
    \label{eq:mm-XE}
\end{equation}
Since both $\|\bm\mu\|^2$ and $\tr\Sigma-\tfrac n2$ are nonnegative (the
latter by \cref{eq:mm-ab} below), $X_E$ is contained in the product
\begin{equation}
    X_E\subseteq M_E\times S_E,
    \qquad
    M_E:=\{\bm\mu:\ \|\bm\mu\|^2\leq E\},
    \qquad
    S_E:=\Big\{\Sigma:\ \tr\Sigma\leq E+\frac n2,\ \Sigma+\frac\ii4\Omega\geq0\Big\}.
    \label{eq:mm-XE-relaxed}
\end{equation}
The set $S_E$ is convex and compact, being the intersection of a linear
matrix inequality with a half-space. The following proposition describes it
explicitly enough for our purposes.

\begin{proposition}
\label{prop:mm-SE-bounds}
Let $\Sigma\in S_E$ with $E\geq0$, and set
\begin{equation}
    \gamma_E:=\tfrac14\big(\sqrt{1+E}-\sqrt E\big)^2,
    \qquad
    \Gamma_E:=\tfrac14\big(\sqrt{1+E}+\sqrt E\big)^2,
    \qquad \gamma_E\Gamma_E=\tfrac1{16},\quad
    \Gamma_E-\gamma_E=\sqrt{E(1+E)} .
    \label{eq:mm-gammaE}
\end{equation}
Then
\begin{enumerate}[label=(\roman*)]
    \item $\gamma_E\,\one_{2n}\leq\Sigma\leq\Gamma_E\,\one_{2n}$;
    \item element-wise, $\gamma_E\leq\Sigma_{kk}\leq\Gamma_E$ for all $k$, and
          $|\Sigma_{kl}|\leq\tfrac12\sqrt{E(1+E)}$ for all $k\neq l$.
\end{enumerate}
\end{proposition}

\begin{proof}
(i) Let $\bm v\in\mathbb R^{2n}$ be a unit vector. Its symplectic partner
$\Omega\bm v$ is also a unit vector, orthogonal to $\bm v$, and
$\bm v^T\Omega(\Omega\bm v)=-1$. Compressing \cref{eq:mm-uncertainty} to
the span of $\bm v$ and $\Omega\bm v$ gives the positive semidefinite
$2\times2$ matrix
$\left(\begin{smallmatrix}a&c\\c&b\end{smallmatrix}\right)
+\tfrac\ii4\left(\begin{smallmatrix}0&-1\\1&0\end{smallmatrix}\right)$,
with $a:=\bm v^T\Sigma\bm v$, $b:=(\Omega\bm v)^T\Sigma(\Omega\bm v)$ and
$c:=\bm v^T\Sigma\,\Omega\bm v$. Its determinant is nonnegative, which gives
the Robertson--Schr\"odinger relation
\begin{equation}
    a\,b\;\geq\;\tfrac1{16}+c^2\;\geq\;\tfrac1{16}.
    \label{eq:mm-ab}
\end{equation}
In particular $a,b>0$, so $\Sigma$ is positive definite, and $a+b\geq\tfrac12$.
Since $U(n)$ acts transitively on the unit sphere of $\mathbb R^{2n}$, the
pair $(\bm v,\Omega\bm v)$ can be completed to an orthonormal symplectic
basis $\{\bm v_j,\Omega\bm v_j\}_{j=1}^n$ with $\bm v_1=\bm v$.
Applying \cref{eq:mm-ab} to each pair gives $a_j+b_j\geq\tfrac12$, and taking
the trace in this basis,
\begin{equation}
    a+b=\tr\Sigma-\sum_{j\geq2}(a_j+b_j)\leq E+\tfrac n2-\tfrac{n-1}{2}
    =E+\tfrac12 .
\end{equation}
For $a,b>0$, the conditions $ab\geq\tfrac1{16}$ and $a+b\leq E+\tfrac12$
imply that $a$ (and $b$) lies between the two roots of
$x^2-(E+\tfrac12)x+\tfrac1{16}=0$, namely
$\tfrac12\big(E+\tfrac12\mp\sqrt{E(1+E)}\big)=\gamma_E,\Gamma_E$.
Hence $\gamma_E\leq\bm v^T\Sigma\bm v\leq\Gamma_E$ for every unit vector
$\bm v$, which is (i).

(ii) The diagonal bounds follow from (i) with $\bm v=\bm e_k$. For $k\neq l$
put $\bm u_\pm=(\bm e_k\pm\bm e_l)/\sqrt2$. Then
$\Sigma_{kl}=\tfrac12(\bm u_+^T\Sigma\bm u_+-\bm u_-^T\Sigma\bm u_-)$, so
$|\Sigma_{kl}|\leq\tfrac12(\Gamma_E-\gamma_E)=\tfrac12\sqrt{E(1+E)}$.
\end{proof}

Together with $M_E \subseteq [-\sqrt{E},\sqrt{E}]^{2n}$, the proposition shows that the set lies in an explicit box,
\begin{equation}
  X_E\subseteq
    [-\sqrt E,\sqrt E]^{2n}\times
    \underbrace{[\gamma_E,\Gamma_E]^{2n}}_{\text{diagonal of }\Sigma}\times
    \underbrace{\big[-\tfrac12\sqrt{E(1+E)},\tfrac12\sqrt{E(1+E)}\big]^{n(2n-1)}}_{\text{off-diagonal of }\Sigma}
    \;\subset\;\mathbb R^{2n^2+3n},
    \label{eq:mm-box}
\end{equation}
in the coordinates $(\bm\mu,\Sigma_{kl})_{k\leq l}$. Moreover, every state in $X_E$
satisfies the uniform lower bound $\Sigma\geq\gamma_E\one_{2n}$. This
bound is what keeps the Wigner function of the considered states and the output densities of Bob's measurement from becoming arbitrarily narrow,
and it is the key input for the classical simulation strategy. For instance, for $n=1$ every
quadrature $\hat q_\varphi=\cos\varphi\,\hat x+\sin\varphi\,\hat p$ has standard
deviation at least $\sqrt{\gamma_E}=\tfrac12(\sqrt{1+E}-\sqrt E)$.

\subsection{General-dyne measurements}

Let $\rho_Z\simeq(\vec0,Z)$, with $Z$ satisfying \cref{eq:mm-uncertainty}, be a Gaussian state (the ``seed''
of a general-dyne measurement). Writing
$\mathcal D(\bm y)=\exp\big[\sum_i(\alpha_ia_i^\dagger-\alpha_i^*a_i)\big]$
for the displacement operator with $\bm y=\tilde{\bm\alpha}\in\mathbb R^{2n}$,
the general-dyne measurement with seed covariance $Z$ is the continuous-outcome
POVM~\cite{Serafini2017,Weedbrook2012}
\begin{equation}
    E_{\bm y}=\frac{1}{\pi^n}\,\mathcal D(\bm y)\,\rho_Z\,
    \mathcal D(\bm y)^\dagger,
    \qquad
    \int_{\mathbb R^{2n}}\mathrm d^{2n}\bm y\;E_{\bm y}=\one .
    \label{eq:mm-generaldyne}
\end{equation}
For a Gaussian state $\rho\simeq(\bm\mu,\Sigma)$, using
\cref{eq:mm-overlap} and the convolution of two normal densities, the
outcome density is
\begin{equation}
    f_Z(\bm y\mid\bm\mu,\Sigma)=\tr(\rho\,E_{\bm y})
    =\mathrm N(\bm y;\bm\mu,\Sigma+Z),
    \qquad \bm y\in\mathbb R^{2n}.
    \label{eq:mm-generaldyne-density}
\end{equation}
For $n=1$ and the pure seed $\rho_Z=\ketbra{0,\zeta}$, \cref{eq:mm-generaldyne} is the POVM \cref{eq:gaussian-generaldyne} of Sec.~\ref{sec:gaussian}, and \cref{eq:mm-generaldyne-density} is \cref{eq: f dens}.
The measurement thus adds the noise $Z$ to the covariance matrix of
the state. We place no restriction on the seed, and will only use $Z\geq0$.

\subsection{A lower bound on multivariate Gaussian densities}
\label{app:mm-lower-bound}

To bound the mass deficit of a cell of states, we show that the output densities $f_Z(\bm y\mid\bm\mu,\Sigma)$ of two Gaussian states are close when their parameters $(\bm \mu, \Sigma)$ are close. More precisely, the output density of a state $(\bm\mu,\Sigma)$ is bounded from below by the density of a reference state $(\bm\mu_0,\Sigma_0)$ times a multiplicative correction that depends simply on the parameter difference. This bound is a property of the multivariate Gaussian distributions in \cref{eq:mm-generaldyne-density}, independent of the rest of the classical simulation model.

Fix the measurement seed $Z$ and write
$\Sigma_Z:=\Sigma+Z$ for the covariance matrix of the output density for a state $(\bm \mu,\Sigma)$, $\Sigma_{0,Z}:=\Sigma_0+Z$ for a reference state $(\bm\mu_0,\Sigma_0)$.  Below, $\|\cdot\|$ denotes the operator norm, $\|\cdot\|_1$ the trace norm and
$|\cdot|$ the Euclidean norm.

\begin{lemma}
\label{lem:mm-density-lower-bound}
Let $Z\geq 0$, let $(\bm \mu, \Sigma)$ and $(\bm \mu_0, \Sigma_0)$ be two displacement vectors and covariance matrices, and let 
\begin{equation}
    \bm\delta:=\bm\mu-\bm\mu_0,\qquad
    \Delta:=\Sigma-\Sigma_0=\Sigma_Z-\Sigma_{0,Z},\qquad
    \bm w:=\Sigma_{0,Z}^{-1/2}(\bm y-\bm\mu_0).
\end{equation}
If $\Sigma_0,\Sigma\geq \gamma \one_{2n}$ and $\epsilon:=\frac{\|\Delta\|}{\gamma}<1$ for some $\gamma>0$, then for all $\bm y\in\mathbb R^{2n}$
\begin{align}
    \mathrm N(\bm y;\bm\mu,\Sigma_Z)\;&\geq\;
    \mathrm N(\bm y;\bm\mu_0,\Sigma_{0,Z})\,
    \exp\!\left[
    -\frac{\|\Delta\|_1}{2\gamma}
    -\frac{\epsilon}{2(1-\epsilon)}\,|\bm w|^2
    -\frac{|\bm\delta|}{\sqrt\gamma\,(1-\epsilon)}\,|\bm w|
    -\frac{|\bm\delta|^2}{2\gamma}
    \right]
    \label{eq:mm-density-lower-bound}
    \\
    &\geq\;
    \mathrm N(\bm y;\bm\mu_0,\Sigma_{0,Z})\,
    \left[1
    -\frac{\|\Delta\|_1}{2\gamma}
    -\frac{\|\Delta\|}{2\gamma (1-\epsilon)}\,|\bm w|^2
    -\frac{|\bm\delta|}{\sqrt\gamma\,(1-\epsilon)}\,|\bm w|
    -\frac{|\bm\delta|^2}{2\gamma}
    \right].
    \label{eq:mm-density-lower-bound-lin}
\end{align}
\end{lemma}

\noindent Note that $\Delta$ and $\bm \delta$ do  not depend on the measurement, while $\bm w$ does not depend on $(\bm \mu, \Sigma)$. In our case $\Sigma,\Sigma_0\geq\gamma_E\one_{2n}$ is guaranteed for all states in $X_E$ by Proposition~\ref{prop:mm-SE-bounds}, while $\epsilon=\frac{\|\Delta\|}{\gamma}<1$ can be guaranteed by taking two states whose covariance matrices are close enough (small enough cell).

\begin{proof}
The second inequality follows from the first by using the definition of $\epsilon$ and the inequality $e^{-x}\geq1-x$, so it suffices to prove the first one. Since $Z\geq0$, the hypothesis $\Sigma,\Sigma_0\geq\gamma\one_{2n}$ implies $\Sigma_Z,\Sigma_{0,Z}\geq\gamma\one_{2n}$, which is what we use below.

Write $\Sigma_Z=\Sigma_{0,Z}^{1/2}(\one+X)\Sigma_{0,Z}^{1/2}$ with
$X:=\Sigma_{0,Z}^{-1/2}\Delta\,\Sigma_{0,Z}^{-1/2}$. Then
$\|X\|\leq\|\Sigma_{0,Z}^{-1}\|\,\|\Delta\|\leq\epsilon<1$ and
$\|X\|_1\leq\|\Delta\|_1/\gamma$.

\emph{Normalization.} Since $\ln(1+x)\leq x$ for all eigenvalues $x>-1$ of
$X$,
\begin{equation}
    \ln\det\Sigma_Z-\ln\det\Sigma_{0,Z}=\ln\det(\one+X)\leq\tr X
    \leq\|X\|_1\leq\frac{\|\Delta\|_1}{\gamma}.
\end{equation}

\emph{Exponent.} Since $\bm y-\bm\mu_0=\Sigma_{0,Z}^{1/2}\bm w$,
\begin{equation}
    (\bm y-\bm\mu)^T\Sigma_Z^{-1}(\bm y-\bm\mu)
    =\bm w^T(\one+X)^{-1}\bm w
    -2\,\bm w^T(\one+X)^{-1}\Sigma_{0,Z}^{-1/2}\bm\delta
    +\bm\delta^T\Sigma_Z^{-1}\bm\delta .
\end{equation}
Since $(\one+X)^{-1}\leq(1-\epsilon)^{-1}\one$, the first term is at most
$|\bm w|^2+\frac{\epsilon}{1-\epsilon}|\bm w|^2$. The second is at most
$\frac{2}{1-\epsilon}|\bm w|\,|\bm\delta|/\sqrt\gamma$, using
$\|\Sigma_{0,Z}^{-1/2}\|\leq\gamma^{-1/2}$. The third is at most
$|\bm\delta|^2/\gamma$. We insert these estimates in
$\ln\mathrm N(\bm y;\bm\mu,\Sigma_Z)
=-\tfrac12(\bm y-\bm\mu)^T\Sigma_Z^{-1}(\bm y-\bm\mu)
-\tfrac12\ln\det(2\pi\Sigma_Z)$ and recognize
$-\tfrac12|\bm w|^2-\tfrac12\ln\det(2\pi\Sigma_{0,Z})
=\ln\mathrm N(\bm y;\bm\mu_0,\Sigma_{0,Z})$. This gives
\cref{eq:mm-density-lower-bound}.
\end{proof}

% In terms of the Euclidean distance, $|\bm w|\leq|\bm y-\bm\mu_0|/\sqrt\gamma$.
% The exponent in \cref{eq:mm-density-lower-bound} is therefore bounded by a
% quadratic polynomial in $|\bm y-\bm\mu_0|$, with coefficients linear in
% $\|\Delta\|$ and $|\bm\delta|$, up to the factor $(1-\epsilon)^{-1}$. We
% keep the whitened variable $\bm w$ because it makes the bound uniform in
% the measurement. For squeezed seeds the variance of $\Sigma_{0,Z}$ diverges in
% some directions, but $\bm w$ remains standard normal.

\subsection{Mass deficit of a cell}
Let $R\subseteq X_E$ be a cell of states and fix a reference
$(\bm\mu_0,\Sigma_0)\in R$ such that $\epsilon = \frac{\|\Sigma-\Sigma_0\|}{\gamma_E}<1$ uniformly over the cell, i.e. $\sup_{(\bm\mu,\Sigma)\in R}\epsilon<1$. By lemma~\ref{lem:mm-density-lower-bound},
the mass deficit of the cell satisfies
 \begin{align}
 \kappa_{\cF}(R) &=1-\inf_{Z} \int \dd^{2n} \bm  y \inf_{(\bm \mu,\Sigma) \in R}  \mathrm N(\bm y;\bm\mu,\Sigma_Z) \\
 &\leq\; \sup_Z \int \dd^{2n} \bm  y \, 
    \mathrm N(\bm y;\bm\mu_0,\Sigma_{0,Z})\,
   \sup_{(\bm \mu,\Sigma) \in R} \left[
    \frac{\|\Delta\|_1}{2\gamma_E}+
    \frac{\|\Delta\|}{2\gamma_E(1-\epsilon)}\,|\bm w|^2+
    \frac{|\bm\delta|}{\sqrt\gamma_E\,(1-\epsilon)}\,|\bm w|+
    \frac{|\bm\delta|^2}{2\gamma_E}
    \right]
 \end{align}
Defining
\begin{equation}
    \delta_R:=\sup_{R}|\bm\mu-\bm\mu_0|,\qquad
    \Delta_R:=\sup_{R}\|\Sigma-\Sigma_0\|,\qquad
    \Delta^{(1)}_R:=\sup_{R}\|\Sigma-\Sigma_0\|_1,\qquad
    \epsilon_R:=\Delta_R/\gamma_E,
\end{equation}
we relax the supremum in the integrand to 
 \begin{align}
 \kappa_{\cF}(R) 
 &\leq\; \sup_Z \int \dd^{2n} \bm  y \, 
    \mathrm N(\bm y;\bm\mu_0,\Sigma_{0,Z})\,
   \left[
    \frac{\Delta_R^{(1)}}{2\gamma_E}+
    \frac{\Delta_R}{2\gamma_E(1-\epsilon_R)}\,|\bm w|^2+
    \frac{\delta_R}{\sqrt\gamma_E\,(1-{\epsilon_R})}\,|\bm w|+
    \frac{\delta_R^2}{2\gamma_E}
    \right] \\
    & =  
    \frac{\Delta_R^{(1)}}{2\gamma_E}+
    \frac{\Delta_R}{2\gamma_E(1-\epsilon_R)}\,{\sup_Z}\, \Exp |\bm w|^2+
    \frac{\delta_R}{\sqrt\gamma_E\,(1-{\epsilon_R})}\,{\sup_Z}\, \Exp|\bm w|+
    \frac{\delta_R^2}{2\gamma_E}
 \end{align}
where in the first line  we used the fact that the vector $ \bm w=\Sigma_{0,Z}^{-1/2}(\bm y-\bm\mu_0)$ does not depend on $(\bm \mu, \Sigma)$ and that the bracket is increasing in $\|\Delta\|_1$, $\|\Delta\|$ and $|\bm\delta|$ (as $\frac{1}{1-\epsilon}$ is increasing for $\epsilon<1$), and replaced the normal integrals with the expected values in the second line.  Under
$\mathrm N(\bm y;\bm\mu_0,\Sigma_{0,Z})$, the vector $\bm w$ is standard normal in
$\mathbb R^{2n}$, so 
\begin{align}
\Exp|\bm w|^2=2n \qquad \text{and}\qquad \Exp|\bm w|\leq\sqrt{2n}
\end{align}
are also independent of $Z$. Hence, we obtain 
\begin{align}
   \kappa_{\cF}(R) 
    &  \;\leq\; {\frac{\Delta^{(1)}_R}{2\gamma_E}}
    +\frac{n\,\Delta_R}{\gamma_E (1-\epsilon_R)}
    +\frac{\sqrt{2n}\,\delta_R}{\sqrt{\gamma_E}\,(1-\epsilon_R)}
    +\frac{\delta_R^2}{2\gamma_E}.
    \label{eq:mm-mass-deficit}
\end{align}
For a small cell $\epsilon_R\ll 1$ and, since $\Delta^{(1)}_R\leq2n\,\Delta_R$, the mass deficit is linear in the widths $\delta_R$ and $\Delta_R$ of the cell, plus a quadratic correction in $\delta_R$. 

Note that homodyne measurements, obtained as limits of infinitely squeezed seeds, are also covered by that bound. Their outcome densities are Gaussian, with covariance given by the sub-block of $\Sigma$ corresponding to the measured quadratures, in a quadrature basis adapted to the measurement. 
Since $\Sigma \geq\gamma_E\one_{2n}$, this sub-block is also bounded from below by $\gamma_E\one_k$, where $k$ is the number of measured quadratures. Hence lemma~\ref{lem:mm-density-lower-bound} and \cref{eq:mm-mass-deficit} give the same upper bound: after marginalizing over the conjugate outcomes, the whitened variable is standard normal in $\mathbb R^k$, so $\Exp|\bm w|^2=k\leq2n$ and $\Exp|\bm w|\leq\sqrt{2n}$.

\subsection{Nested partitions and exact simulation}
\label{app:mm-partitions}

The box \cref{eq:mm-box} provides the nested partitions. We bisect each of
its $2n^2+3n$ sides $r$ times, intersect the resulting cubes with $X_E$, and
discard the empty cells. The cells at level $r$ have side length
\begin{equation}
    h_r\leq2^{-r}h_0,\qquad h_0:=\max\big\{2\sqrt E,\sqrt{E(1+E)}\big\}.
\end{equation}
Within a cell $R$, take any reference point $(\bm\mu_0,\Sigma_0)\in R$.
For every other point $(\bm\mu,\Sigma)\in R$, each of the $2n$ components of $\bm\delta=\bm\mu-\bm\mu_0$ and each of the $(2n)^2$ entries of the symmetric matrix $\Delta=\Sigma-\Sigma_0$ is at most $h_r$ in absolute value. Hence $|\bm\delta|\leq\sqrt{2n}\,h_r$, and the Frobenius norm $\|\Delta\|_F:=\big(\sum_{k,l}\Delta_{kl}^2\big)^{1/2}$ satisfies $\|\Delta\|_F\leq2n\,h_r$. Using the standard inequalities $\|\Delta\|\leq\|\Delta\|_F$ and $\|\Delta\|_1\leq\sqrt{2n}\,\|\Delta\|_F$ for $2n\times2n$ matrices~\cite{HornJohnson2012} and taking the supremum over the cell gives
\begin{equation}
    \delta_R\leq\sqrt{2n}\,h_r,\qquad
    \Delta_R\leq2n\,h_r,\qquad
    \Delta^{(1)}_R\leq(2n)^{3/2}h_r .
\end{equation}
For $h_r\leq\gamma_E/(4n)$ we have $\epsilon_R\leq\tfrac12$, and
\cref{eq:mm-mass-deficit} gives
\begin{equation}
    \kappa_{\cF}(\cC_r)\leq C_{n,E}\,h_r,
    \qquad
    C_{n,E}:=\frac{\sqrt2\,n^{3/2}+4n^2}{\gamma_E}+\frac{4n}{\sqrt{\gamma_E}}+\frac14 ,
    \label{eq:mm-kappa-bound}
\end{equation}
uniformly over all general-dyne and homodyne measurements. For $h_r>\gamma_E/(4n)$, which includes the first levels, the condition $\epsilon_R<1$ may fail and we do not use \cref{lem:mm-density-lower-bound}. The bound \cref{eq:mm-kappa-bound} nevertheless holds trivially: the mass deficit never exceeds $1$, while the right-hand side satisfies $C_{n,E}h_r\geq\frac{4n^2}{\gamma_E}\cdot\frac{\gamma_E}{4n}=n\geq1$. Hence for all $r\geq0$
\begin{equation}
    |\cC_r|\leq2^{(2n^2+3n)r},\qquad
    \kappa_{\cF}(\cC_r)\leq K_{n,E}\,2^{-r},\qquad K_{n,E}=C_{n,E}\,h_0,
\end{equation}
in agreement with \cref{eq:gaussian-mass-deficit-constant}. Let
$N_{n,E}:=\lceil\log\max\{1,K_{n,E}\}\rceil$. At shared level $\lambda=j$, Alice
uses the partition $\cC_{N_{n,E}+j}$, whose mass deficit is at most
$2^{-j}$, and her label requires at most $(2n^2+3n)(N_{n,E}+j)$ bits.
Since $\Exp[\lambda]=2$, this proves \cref{eq:gaussian-average-bits}, and for $n=1$ \cref{eq:gaussian-average-bits-1mode}. In particular, the partitions $\cC_r$ show that $d_{\mathrm{mdM}}(X_E,\cF)\leq2n^2+3n$, as claimed in Sec.~\ref{sec:gaussian}.

\end{widetext}

\end{document}